\documentclass[12pt, oneside]{amsart}
\pdfoutput=1
\usepackage{geometry}
\usepackage{lmodern}
\usepackage{amsmath}
\usepackage{amssymb}
\usepackage{amsthm}
\usepackage{comment}
\usepackage{color}
\usepackage[usenames,dvipsnames]{xcolor}
\usepackage{graphicx}
\usepackage{subfig}
\usepackage{soul}
\usepackage[countmax]{subfloat}
\usepackage{float}
\usepackage{rotfloat}
\usepackage{setspace}
\usepackage{esint}
\usepackage[authoryear]{natbib}
\definecolor{MyDarkBlue}{rgb}{0,0.08,0.45}
\usepackage[unicode=true,pdfusetitle, bookmarks=true,bookmarksnumbered=false,bookmarksopen=false, breaklinks=false,pdfborder={0 0 1},backref=section,colorlinks=true, linkcolor = MyDarkBlue, citecolor = MyDarkBlue]{hyperref}
\usepackage{breakurl}
\usepackage[para]{threeparttable}

\PassOptionsToPackage{normalem}{ulem}
\usepackage{ulem}

\makeatletter

\numberwithin{equation}{section}

\hypersetup{colorlinks=true, linkcolor=MyDarkBlue}

\newtheorem{theorem}{{\bf\sc Theorem}}

\newtheorem{thm}{{\bf \sc Theorem}}[section]

\newtheorem{cor}{{\bf\sc Corollary}}[section]

\providecommand{\E}{\mathrm{E}}
\providecommand{\var}{\mathrm{var}}

\providecommand{\cov}{\mathrm{cov}}

\providecommand{\g}{{\bf g}}
\providecommand{\bfH}{{\bf H}}

\renewenvironment{proof}[1][Proof]{\noindent\text{#1.} }{\ \rule{0.5em}{0.5em}}

\makeatother

\title[Using Gossips to Spread Information]{Using Gossips to Spread Information: \\ Theory and Evidence from a Randomized Controlled Trial
}

\author{Abhijit Banerjee$^{\dagger}$ }
\author{Arun G. Chandrasekhar$^{\ddagger}$}
\author{Esther Duflo$^{\S}$ }
\author{Matthew O. Jackson$^{\star}$ }
\date{This Version: \today}

\thanks{This paper supersedes an earlier paper ``Gossip: Identifying Central Individuals in a Social Network.''
Financial support from the NSF under grants
SES-1156182, SES-1155302, and SES-1629446, and from the AFOSR and DARPA under grant
FA9550-12-1-0411, and from ARO MURI under award No. W911NF-12-1-0509
is gratefully acknowledged. We thank Shobha Dundi, Devika Lakhote,
Tithee Mukhopadhyay, and Gowri Nagraj for outstanding research assistance.
We also thank Michael Dickstein, Ben Golub, John Moore, and participants at various seminars/conferences
for helpful comments.
Social Science Registry AEARCTR-0001770 and approved by MIT IRB COUHES \# 1010004040.
}
\thanks{$^{\dagger}$Department of Economics, MIT}
\thanks{$^{\ddagger}$Department of Economics, Stanford University}
\thanks{$^{\S}$Department of Economics, MIT}
\thanks{$^{\star}$Department of Economics, Stanford University; Santa Fe Institute; and CIFAR}

\begin{document}

\begin{titlepage}
\maketitle
\begin{abstract}
Is it possible to identify individuals who are highly central in a
community without gathering any network information, simply by asking
a few people? If we use people's nominees as seeds for a diffusion
process, will it be successful? We explore these questions theoretically,
via surveys, and via field experiments. We show via a model of information
flow how members of a community can, just by tracking gossip about
others, identify highly central individuals in their network. Asking
villagers in rural Indian villages to name good seeds for diffusion,
we find that they accurately nominate those who are central according
to a measure tailored for diffusion \textendash{} not just those with
many friends or in powerful positions. Finally, we run a randomized
field experiment in 213 other villages that tests how effective it
is to use such nominations as seeds for a diffusion process. Relative
to random seeds or those with high social status, hitting at least
one seed nominated by villagers leads to more than a 65\% increase
in the spread of information.

\textsc{JEL Classification Codes:} D85, D13, L14, O12, Z13

\textsc{Keywords:} Centrality, Gossip, Networks, Diffusion, Influence, Social Learning
\end{abstract}

\end{titlepage}

\section{Introduction}

\begin{quote}
\emph{``The secret of my influence has always been that it remained
secret.''} \textendash{} Salvador Dal\'{i}
\end{quote}
\

To diffuse information by word-of-mouth and influence behavior in
a community, one needs to identify and seed the information via central
individuals.\footnote{See \citet*{katzl1955,rogers1995,kempekt2003,kempekt2005,borgatti2005,ballestercz2006,banerjeecdj2013}.}
Moreover, as shown in \citet*{banerjeecdj2013} and \citet*{beaman2014can},
even though many measures of centrality are correlated, successful
diffusion requires seeding information via people who are central according
to specific measures. A practical challenge is that the relevant centrality
measures are based on extensive network information, which can be
costly and time consuming to collect in many settings.

In this paper, we thus ask and answer a question that has never been
examined before: How can one easily and cheaply identify highly central
individuals without gathering network data? As shown in the work mentioned
above, superficially obvious ``fixes'' for finding central individuals
\textendash{} such as targeting people with leadership or special
status, or who are geographically central, or even those with many
friends \textendash{}
can fail when it comes to diffusing information.
So, how can one find highly central individuals without network data
and in ways that are more effective than relying on status labels?
We explore a direct technique that turns out to be remarkably effective:
simply asking a few individuals in the population who would be the
best individuals for spreading information.

There is ample reason to doubt that such a technique would work. Previous
studies have shown that people's knowledge about the networks in which
they are embedded is surprisingly lacking. In fact, individuals within
a network tend to have little perspective on its structure, as found
in important research by \citet*{friedkin1983} and \citet*{krackhardt1987},
among others.\footnote{See \citet*{krackhardt2014} for background and references.}
This raises the question of whether and how, despite not knowing the
structure of the network in which they are embedded, people know who
is central and well-placed to diffuse information through the network.

In this paper, we examine people's ability to identify highly central
individuals and effective seeds for a diffusion process. We make three
main contributions.

Our first contribution is theoretical. We answer the question of how
it could possibly be that people could name highly diffusion central
individuals without knowing anything about their network. Because
we are interested in diffusion, the feature of the network that we
hope people would have knowledge about is a notion of centrality which
relates to iterative expansion properties of the social network (which
we have defined as ``diffusion centrality''). Needless to say, this
is a complicated concept, and so superficially it may seem implausible
that people could estimate it. Our main theoretical result shows
that there is a very simple argument for why even very naive agents, simply by counting
how often they hear pieces of gossip, would have accurate estimates
of others' centralities. Of course, this is just a possibility result: it is obviously not the only possible way in which people
may learn who is central.\footnote{This result demonstrates what is special about
this notion of centrality. Even without any knowledge of the network,
this is information that individuals can easily get and process.}

To show that individuals \emph{can} learn to identify central individuals
within their community even \emph{without knowing anything about the
structure of their network}, we model a process that we call ``gossip,''
in which nodes generate pieces of information that are stochastically
passed from neighbor to neighbor, along with the identity of the node
from which the information emanated. We assume only that individuals
who hear the gossip are able to keep count of the number of times
that each person in the network is mentioned as a source.\footnote{We use the
term ``gossip'' to refer to the spreading of information
about particular people. Our diffusion process is focused on basic
information that is not subject to the biases or manipulations that
might accompany some ``rumors'' (e.g., see \citet*{blochdk2014}).}
We show that for any listener in the network, the relative ranking
under this count converges over time to the correct ranking of every
node's centrality.\footnote{The specific definition of centrality
we use here is diffusion centrality
\citep{banerjeecdj2013} but a similar result
holds for eigenvector centrality, as is seen in the Appendix.}

Our second contribution is to show that when we ask people in the
field to name good diffusers the nominees are indeed highly central.
To this end, we collected data on who villagers think would be good
at spreading information in 33 villages where we had previously collected
detailed network data. We asked every adult to name the person in
their village best suited to initiate the spread of information. We
show that, indeed, individuals nominate highly diffusion central people.
Nominees consistently rank in the top quartile of centrality, and
many rank in the top decile. We also show that the nominations are
not simply based on the nominee's leadership status or geographic
position in the village, but are significantly correlated with diffusion
centrality even after controlling for these characteristics.

Thus we have reason to believe, both theoretically and empirically,
that people are good at naming highly diffusion-central individuals.
Despite this correlation, it remains possible that individuals identified
by others as good at transmitting information would in fact not be
good seeds.

Our third contribution is then to field-test our method of cheaply
identifying seeds and diffusing information.
To do this we ran a new experiment in 213 new villages not previously
in any of our samples. We asked villagers
who would be a good diffuser of information. In 71 of those villages,
we then used those nominations
to seed information about a give-away of free cell phones. We compare
how well these nominated seeds do compared to another 71 villages
in which we selected seeds who villagers reckon to have high social status,
and yet another 71 villages in which we selected the seeds randomly.

Specifically, in the random seeding villages we seeded a piece information
in 3 to 5 randomly selected households (the number of seeds to be
reached was randomly selected). In the ``social status'' villages,
we seeded information in 3 to 5 village households who have status as
``elders'' in the village \textendash{} leaders with a degree of
authority in the community, who command respect. In the remaining
71 villages, we seeded information in 3 to 5 individuals nominated by
others as being well suited to spread information (``gossip nominees'').
The piece of information that we spread is simple: anyone who calls
a particular phone number will have a chance to win a free cell phone,
and if they do not win the phone, they are guaranteed to win some
cash. The chances to win cash and phones are independent of the number
of people who respond, ensuring that the information is non-rivalrous
and everyone was informed of that fact. The call itself is free. We
then measure the extent of diffusion using the number of independent
entrants.

We received on average 8.1 phone calls in villages with random seedings,
6.9 phone calls in villages with village elder seedings, and 11.7
in villages with gossip seedings. The additional 50 percent participation
rate from using gossip compared to random seedings is the relevant
difference for a policy maker considering whether to use a technique
of asking for ``gossip'' nominations to seed a piece of information.
We also estimate the impact of seeding with ``gossip nominations'':
in quite a few villages with random seedings, a gossip nominee was
hit by chance. We can thus measure how much better ``gossip seeds''
are at circulating information than other seeds. We find that in villages
where no gossip (and no elders) were seeded, we received only 5.8
calls. In villages where at least one gossip was seeded, we received
3.78 more calls, a 65\% increase. If we instrument ``hitting at least
one gossip'' with the gossip treatment, we find a similar result:
seeding at least one gossip seed yields an extra 7.4 calls.

Thus, although call-back rates are moderate, we get about twice as
many entries when we seed information with gossip nominees as compared
to seeding with village elders or with random non-nominated villagers.

To test whether the increase in diffusion from gossip nominees is
in fact accounted for by their diffusion centrality, we went back
to most of the villages with random seeding just after the experiment,
and collected full network data. Consistent with network theory, we
find that information diffuses faster when we hit at least one seed
with high diffusion centrality. However, when we include both gossip
nomination and diffusion centrality of the seeds in the regression,
the coefficient of gossip centrality does not decline much (although
it becomes less precise). This suggests that diffusion centrality
does not explain all of the extra diffusion from gossip nominees.
People's nominations may incorporate additional attributes, such as
who is listened to in the village, or who is most charismatic or talkative,
etc., which goes beyond a nominee's centrality. Alternatively, it
may be that our measure of the network and diffusion centrality are
noisy, and villagers are even more accurate at finding central individuals
than we are.

To summarize, we suggest a process by which, by listening
and keeping count of how often they hear \emph{about} someone, individuals
learn the correct ranking of community members in terms of how effectively
they can spread information.  And, we show that, in practice, individuals
nominated by others are indeed effective seeds of information.

\subsection*{Contribution and Relation to the literature}

There is a voluminous literature on the role of opinion leaders and
key individuals in diffusing products and information. This ranges
from the early sociology literature (e.g, classic studies by \citet*{simmel1908,katzl1955,colemankm1966}),
to the vast literature on diffusion of innovations (e.g., \citet*{rogers1995,centola2010,centola2011,jacksony2011}),
to appropriate measures of centrality (e.g., \citet*{bonacich1987,borgatti2005,ballestercz2006,valente2008,limot2015,blochjt2016}),
to a literature on identifying central and influential individuals
in marketing (e.g., \citet{krackhardt1996,iyengarvv2010,hinzetal2011,katonazs2011}),
to the computational issues of identifying multiple individuals for
seeding (e.g., \citep*{kempekt2003,kempekt2005}).

To our knowledge, this is the first paper to demonstrate that members
of communities are able, easily and accurately, to nominate people
in the community who are good at diffusing information, and that these
nominees are highly central in a network sense. It is important to
emphasize that this is very distinct from using the friendship paradox
(\citet{feld1991}) to find high-degree individuals. That is, since
high-degree individuals have more friends than low-degree people,
a standard way of finding high-degree individuals is simply to ask
people to name their friends (e.g., see \citet{krackhardt1996,kimetal2015,jackson2016}).
Here, we are trying to find people who are central in terms of measures
that are more complex than degree-centrality, and the theory and techniques
we develop are correspondingly different from the standard approaches
in viral marketing.

Our work is also the first to describe a simple process by which people
can learn things about their broader network to which they have no
direct access.\footnote{There are some papers (e.g., \citet{milgram1967} and \citet{doddsmw2003})
that have checked people's abilities to use knowledge of their friends'
connections to efficiently route messages to reach distant people;
those papers, however, test knowledge about peoples' own connections.} Our results have important practical consequences, since policy makers
and businesses are often looking for the best way to spread information,
and asking people to identify the best person to spread the information
is cheaper and easier than collecting detailed network data.

The centrality measure that we work with is the notion of ``diffusion
centrality'' that we defined in \citep*{banerjeecdj2013} and found
to significantly predict successful diffusion seedings. Diffusion
centrality measures how widely information from a given node diffuses
in a given number of time periods and for a given random per-period
transmission probability. In supporting materials, we prove that this measure
of centrality nests three of the most prominent measures in the literature:
degree centrality at one extreme (if there is just one time period
of communication), and eigenvector centrality and Katz\textendash Bonacich
centrality at the other extreme (if there are unlimited periods of
communication). For intermediate numbers of periods, diffusion centrality
takes on a range of other values.

It is important to emphasize that even though we work with a centrality
measure that we defined in earlier work, our perspective in this paper
is completely different. 
Our previous work explored which measure of centrality best predicts
diffusion, and whether microfinance participation depends on peer
endorsement. 
Here we examine how we can leverage people's knowledge to identify
highly central individuals, without relying on any network information.
This question is important for at least two reasons: for developing
cost-effective methods of diffusing information and for understanding
what people know about the networks in which they are embedded.

There are two limitations that are worth highlighting and discussing.
First this paper focuses on the pure transmission of information -
simple knowledge that is either known or not. In some applications,
people may not only need to know of an opportunity but may also be
unsure of whether they wish to take advantage of that opportunity,
and thus may also rely on endorsements of others. In those cases,
trust in the sender will also matter in the diffusion process. We
focus, for most of the paper, on the spread of simple sorts of information,
and in the experiment, the piece of information we seeded was designed
not to require trust in order to participate. Although issues of trust
are certainly relevant in some applications, pure lack of information
is often a binding and important constraint, and is therefore worthy
of study. In addition, in our work on microfinance \citep*{banerjeecdj2013},
for example, we could not reject the hypothesis that the role of the
social network in the take up of microfinance was entirely mediated
by information transmission, and that endorsement played no role.

Second, our experiments here are limited to communities on the order
of a thousand people. It is clear that peoples' abilities to name
highly central individuals may not scale fully to networks that involve
hundreds of thousands or millions of people. Nonetheless, our work
still demonstrates that people are effective at naming central people
within reasonably sized communities. There are many settings, in both
the developing and developed world, in which person-to-person communication
within a community, company, department, or organization of limited
scale is important. Our model and empirical findings are therefore
a useful first step in a broader research agenda.\footnote{More generally,
one may want to choose many seeds in a large society, some within in each of various
sub-communities, in which case the techniques developed here would still be useful.}

The remainder of the paper is organized as follows: Section \ref{sec:model}
develops our model of diffusion. In Section \ref{sec:gossip}, we
relate the notion of diffusion centrality to network gossip. Section
\ref{sec:data} describes the setting and the data used in the empirical
analysis. We examine whether individuals nominate central nodes in
Section \ref{sec:results}. In Section \ref{sec:experiment}, we describe
the field experiment and results. Section \ref{sec:conclusion} concludes.

\

\section{A Model of Network Communication}\label{sec:model}

We consider the following model.

\subsection{A Network of Individuals}

\

A society of $n$ individuals are connected via a possibly directed
and weighted network, which has an adjacency matrix ${\bf g}\in[0,1]^{n\times n}$.\footnote{When defining $\g$ in the directed case, the $ij$-th entry
indicates that $i$ can tell something to $j$. In some networks, this
may not be reciprocal. }
Unless otherwise stated, we take the network $\g$ to be fixed and
let $v^{(R,1)}$ be its first (right-hand) eigenvector, corresponding
to the largest eigenvalue $\lambda_{1}$.\footnote{$v^{(R,1)}$ is such that ${\bf g}{v}^{(R,1)}=\lambda_{1}{v}^{(R,1)}$
where $\lambda_{1}$ is the largest eigenvalue of ${\bf g}$ in magnitude.} The first eigenvector is nonnegative and real-valued by the Perron\textendash Frobenius
Theorem. Throughout what follows, we assume that the network is (strongly)
connected in that there exists a (directed) path from every node to
every other node, so that information originating at any node could
potentially make its way eventually to any other node.\footnote{More generally, everything that we say applies to components of the
network.}

Two concepts, \emph{diffusion centrality and network gossip }will
be central to the theory developed here. We introduce them one by
one and then show how they are connected.

\subsection{Diffusion Centrality}

\label{sec:properties}

\

In \citet*{banerjeecdj2013}, we defined a notion of centrality called
\emph{diffusion centrality}, based on random information flow through
a network according to the following process, which is a variant of
the standard process that underlies many models of contagion.\footnote{See \citet{jacksony2011} for background and references. A continuous
time version of diffusion centrality was subsequently defined in \citet{lawyer2014}.}


A piece of information is initiated at node $i$ and then broadcast
outwards from that node. In each period, with probability $q\in(0,1]$,
independently across neighbors and history, each informed node informs
each of its neighbors of the piece of information and the identity
of its original source.\footnote{Note that since we allow $\g$ to be a fully heterogeneous matrix
(a weighted and directed graph), $q$ is redundant. However, for the
purposes of relating the theory to our empirical work, it is useful
to think of $\g$ as an unweighted graph, since survey network data
often just indicates whether households have a connection. For this
reason, we include the $q$-parameter explicitly, as it will be relevant
for our empirical exercise and also lead to new insights about how
diffusion centrality behaves as the communication rate varies. But
in the Appendix and online materials,
all of our results and proofs allow for an arbitrary
weighted and directed graph, and thus full heterogeneity in the probability
that two nodes interact, in which case $q$ is obviously redundant.} The process operates for $T$ periods, where $T$ is a positive integer.

We emphasize that there are good reasons to allow $T$ to be finite.
For instance, a new piece of information may only be relevant for
a limited time. Also, after some time, boredom may set in or some
other news may arrive and the topic of conversation may change.

Diffusion centrality measures how extensively the information spreads
as a function of the initial node. In particular, let
\[
{\bf H}({\bf g};q,T):=\sum_{t=1}^{T}\left(q{\bf g}\right)^{t},
\]
be the ``hearing matrix.'' The $ij$-th entry of ${\bf H}$, $H({\bf g};q,T)_{ij}$,
is the expected number of times, within $T$ periods, that $j$ hears
about a piece of information originating from $i$. Diffusion centrality
is then defined by
\[
{DC}({\bf g};q,T):={\bf H}({\bf g};q,T)\cdot{\bf 1}=\left(\sum_{t=1}^{T}\left(q{\bf g}\right)^{t}\right)\cdot{\bf 1}.
\]

So, ${DC}({\bf g};q,T)_{i}$ is the expected total number of times
that some piece of information that originates from $i$ is heard
by any of the members of the society during a $T$-period time interval.\footnote{We note two useful normalizations. One is to compare this calculation
to what would happen if $q=1$ and ${\bf g}$ were the complete network
${\bf g}^{c}$, which produces the maximum possible entry for each
$ij$ for any given any $T$. Thus, each entry of ${DC}({\bf g};q,T)$
could be divided through by the corresponding entry of ${DC}({\bf g}^{c};1,T)$.
This produces a measure for which every entry lies between 0 and 1,
where 1 corresponds to the maximum possible number of expected walks
possible in $T$ periods with full probability weight and full connectedness.
Another normalization is to compare a given node to the total level
for all nodes; that is, to divide all entries of ${DC}({\bf g};q,T)$
by $\sum_{i}{DC}_{i}({\bf g};q,T)$. This normalization tracks how
relatively diffusive one node is compared to the average diffusiveness
in its society.} \citet{banerjeecdj2013} showed that diffusion centrality of the
initially informed members of a community was a statistically significant
predictor of the spread of information \textendash{} in that case,
about a microfinance program.

It is useful to remind the reader of diffusion centrality's relationship
to other prominent measures of centrality, though a reader impatient
to see our main results is welcome to bypass this.

As we stated in \citet{banerjeecdj2013}, for different values of
$T$, diffusion centrality nests three of the most prominent and widely
used centrality measures: degree centrality, eigenvector centrality,
and Katz\textendash Bonacich centrality.\footnote{Let $d({\bf g})$
denote (out) degree centrality: $d_{i}({\bf g})=\sum_{j}g_{ij}$.
Eigenvector centrality corresponds to $v^{(R,1)}({\bf g})$: the first
eigenvector of ${\bf g}$. Also, let $KB({\bf g},q)$ denote Katz\textendash Bonacich
centrality \textendash{} defined for $q<1/\lambda_{1}$ by $KB({\bf g},q):=\left(\sum_{t=1}^{\infty}\left(q{\bf g}\right)^{t}\right)\cdot{\bf 1}.$} It thus provides a foundation for these measures and spans the gap
between them.

In particular, it is straightforward to show that (i) diffusion centrality
is proportional to (out) degree centrality at the extreme at which
$T=1$, and (ii) if $q<1/\lambda_{1}$, then diffusion centrality
coincides with Katz\textendash Bonacich centrality if we set $T=\infty$.
It takes more work to show that, when $q>1/\lambda_{1}$, diffusion
centrality approaches eigenvector centrality as $T$ approaches $\infty$.
Intuitively, the difference between the extremes of Katz\textendash Bonacich
centrality and eigenvector centrality depends on whether $q$ is sufficiently
small so that limited diffusion takes place even for large $T$, or
whether $q$ is sufficiently large so that the knowledge saturates
the network and then it is relative amounts of saturation that
are captured by this measure.

The exact threshold makes sense:
whether $q$ is above or below $1/\lambda_{1}$ determines whether
the sum in diffusion centrality converges or diverges \textendash{}
and, as we know from spectral theory the first eigenvalue of a matrix
governs its expansion properties. For completeness, a formal statement
and proof of these results appears in the Appendix \ref{sec:proofs}.

Interestingly, the same threshold for $q$ plays an important role
even when $T$ is finite. In the Appendix, we provide new theoretical
results on diffusion centrality that show that diffusion centrality
behaves fundamentally differently depending on whether $q$ is above
or below $1/\lambda_{1}$ for reasons similar to those discussed
already. We also show that diffusion centrality behaves quite differently
depending on whether $T$ is smaller or bigger than the diameter of the graph.
The reason is that in many large graphs, the average
distance between most nodes is actually almost the same as the diameter,
something first discovered by Erdos and Renyi. Thus, if $T$ is below
the diameter, news from any typical node will not have a long enough
time to reach most other nodes. In contrast, once $T$ hits the diameter,
then that permits news from any typical node to reach most others.
When $T$ exceeds the diameter of the graph, then many of the walks counted
by ${\bf g}^{T}$ begin to have ``echoes'' in them: they visit some
nodes twice. For instance, news passing from node 1 to node 2 to node
3 then back to node 2 and then to node 4, etc. Once most walks have
echoes in them, the measure begins to act differently,
and the diffusion centrality vector eventually converges
to the ergodic distribution, and essentially  the first
eigenvector (provided $q$ is large enough to get saturation).

These results are formally proved in the Appendix \ref{sec:Echoes}. From the point
of view of the empirical exercises that are at the heart of this paper,
these results are very useful because they suggest that the threshold
case of $q=1/\E[\lambda_{1}]$ and $T=\E[Diam(\g)]$ provides a natural
benchmark value for $q$ and $T$. This allows us to assign numerical
values to ${DC}({\bf g};q,T)_{i}$.

\subsection{Network Gossip}

Diffusion centrality considers diffusion from the \emph{sender's}
perspective. Let us now consider the same information diffusion process
but from a \emph{receiver's} perspective. Over time, each individual
hears information that originates from different sources in the network,
and in turn passes that information on with some probability. The
society discusses each of these pieces of information for $T$ periods.
The key point is that there are many such topics of conversation,
originating from all of the different individuals in the society,
with each topic being passed along for $T$ periods.

For instance, Arun may tell Matt that he has a new car. Matt may then
tell Abhijit that ``Arun has a new car,'' and then Abhijit may tell
Esther that ``Arun has a new car.'' Arun may also have told Ben
that he thinks house prices will go up, and Ben could have told Esther
that ``Arun thinks that house prices will go up.'' In this model,
Esther keeps track of the cumulative number of times bits of information
that originated from Arun reach her and compares it with the number
of times she hears bits of information that originated from other
people. What is crucial, therefore, is that the news involves the
name of the node of origin \textendash{} in this case ``Arun'' \textendash{}
and not what the information is about. The first piece of news originating
from Arun could be about something he has done (``bought a car''),
but the second could just be an opinion (``Arun thinks house prices
will go up''). Esther keeps track of how often she hears of things
originating from Arun. Then Esther imputes peoples' centralities based
on how often she hears about them. She estimates Abhijit's, Arun's,
Ben's, Matt's, ..., Sara's centralities just based on the frequency
that she hears things that originated at each one of them.\footnote{Of course, one can imagine all kinds of gossip processes and could
enrich the model along many dimensions. The point here is simply to
provide a ``possibility'' result - to understand how it could be
that people can easily learn information about the centrality of others.
Noising up the model could noise up people's knowledge of others'
centralities, but this benchmark gives us a starting point.}

Recall that
\[
{\bf H}({\bf g};q,T)=\sum_{t=1}^{T}\left(q{\bf g}\right)^{t},
\]
is such that the $ij$-th entry, $H({\bf g};q,T)_{ij}$, is the expected
number of times $j$ hears a piece of information originating from
$i$.

We define the \emph{network gossip heard} by node $j$ to be the $j$-th
column of ${\bf H}$,
\[
NG({\bf g};q,T)_{j}:=H({\bf g};q,T)_{\cdot j}.
\]
Thus, $NG_{j}$ lists the expected number of times a node $j$ will
hear a given piece of news as a function of the node of origin of
the information. So, if $NG({\bf g};q,T)_{ij}$ is twice as high as
$NG({\bf g};q,T)_{kj}$ then $j$ is expected to hear news twice as
often that originated at node $i$ compared to node $k$, presuming
equal rates of news originating at $i$ and $k$.

Note the different perspectives of $DC$ and $NG$: diffusion centrality
tracks how well information spreads from a given node, while network
gossip tracks relatively how often a given node hears information
from (or about) each of the other nodes.

\section{Relating Diffusion Centrality to Network Gossip}\label{sec:gossip}

We now turn to the first of our main results.

We first investigate whether and how individuals living in network
${\bf g}$ can end up with knowledge of other peoples' diffusion centralities,
without knowing anything about the network structure.

.

\subsection{Identifying Central Individuals}

\

With these two measures of diffusion centrality and network gossip in
hand, we show how individuals in a society can estimate who is central
simply by counting how often they hear gossip that originated with
others. We first show that, on average, individuals' rankings of others
based on $NG_{j}$, the amount of gossip that $j$ has heard about
others, is positively correlated with diffusion centrality for any
$q,T$.

\begin{theorem} \label{prop:cov} For any $({\bf g};q,T)$,
\[
\sum_{j}\cov(DC({\bf g};q,T),NG({\bf g};q,T)_{j})=\var(DC({\bf g};q,T)).
\]
Thus, in any network with differences in diffusion centrality among
individuals, the average covariance between diffusion centrality and
network gossip is positive. \end{theorem}

It is important to emphasize that although both measures, network
gossip and diffusion centrality, are based on the same sort of information
process, they are really two quite different objects. Diffusion centrality
is a gauge of a node's ability to \textsl{send} information, while
the network gossip measure tracks the \textsl{reception} of information
about different nodes. Indeed, the reason that Theorem \ref{prop:cov}
is only stated for the sum, rather than any particular individual
$j$'s network gossip measure, is that for small $T$ it is possible
that some nodes have not even heard about other relatively distant
nodes, and moreover, they might be biased towards their local neighborhoods.\footnote{ One might conjecture that more central nodes would be better ``listeners'':
for instance, having more accurate rankings than less central listeners
after a small number of periods. Although this might happen in some
networks, and for many comparisons, it is not guaranteed. None of
the centrality measures considered here ensure that a given node,
even the most central node, is positioned in a way to ``listen''
uniformly better than all other less central nodes. Typically, even
a most central node might be farther than some less central node from
some other important nodes. This can lead a less central node to hear
some things before even the most central node, and thus to have a
clearer ranking of at least some of the network before the most central
node. Thus, for small $T$, the $\sum$ is important in Theorem \ref{prop:cov}. }

Next, we show that if individuals exchange gossip over extended periods
of time, every individual in the network is eventually able to \emph{perfectly}
rank others' centralities \textendash{} not just ordinally, but cardinally.

\begin{theorem} \label{prop:rank} If $q\geq1/\lambda_{1}$ and $\g$
is aperiodic, then as $T\rightarrow\infty$ every individual $j$'s
ranking of others under $NG({\bf g};q,T)_{j}$ converges to be proportional
to diffusion centrality, $DC({\bf g};q,T)$, and hence according to
eigenvector centrality, $v^{(R,1)}$. \end{theorem}

The intuition is that individuals hear (exponentially) more often
about those who are more diffusion/eigenvector central, as the number
of rounds of communication tends to infinity. Hence, in the limit,
they assess the rankings according to diffusion/eigenvector centrality
correctly. The result implies that even with very little computational
ability beyond remembering counts and adding to them, agents can come
to learn arbitrarily accurately complex measures of the centrality
of everyone in the network, including those with whom they do not
associate.

It is worth emphasizing that although in the above we do not allow
the probability of transmitting information to depend on the person,
this is only for presentation purposes. As noted above, in the Appendix
we work with a fully weighted and directed graph $\g$ which allows
the probability of transmission to vary arbitrarily by pair. Our results
still hold exactly, substituting a condition on the first eigenvalue
of ${\bf g}$ being bigger than 1 in place of the comparison between
$q$ and $1/\lambda_{1}$ in the case presented here. Thus, people
who are effective (or whose friends are effective) at communicating
information would be heard about a lot, and information communicated
to them would also circulate effectively - being fully accounted for
both in diffusion centrality and the network gossip measure.

It is possible that more sophisticated strategies where individuals
try to infer network topology, could accelerate learning. Nonetheless,
what our result underscores is that learning is possible even in an
environment where individuals do not know the structure of the network
and do not tag anything but the source of the information.

The restriction to $q\geq1/\lambda_{1}$ is important. When $q$ tends
to 0, individuals hear about others in the network with vanishing
frequency, and as a result, the network distance between people can
influence who they think is the most important.

Also, nodes are similar in how frequently they generate new
information or gossip. However,
provided the generation rate of new information is positively related
to nodes' centralities, the results that we present below still hold
(and, in fact, the speed of convergence could increase), though
of course if the rate of generation of information about nodes is
negatively correlated with their position, then our results below
would be attenuated. Regardless, the result is still of interest.

We also rule out hearing about people in other ways than through communication
with friends: information only travels through edges in the network.
This is realistic in the contexts we study. Note, however, that things
like media outlets are easily treated as nodes in the network, especially
given that our analysis allows for arbitrarily weighted and directed networks.

Also, agents do not doubt the quality of information, either
in the gossip process or in the information transmission process:
there is no notion of trust, or endorsement. It could be, for example,
that gossips are people who love to talk but are not necessarily reliable.
In that case, their friends may resist passing on information originating
from them even though they themselves may be much talked about. This
could be of interest in some settings and is an interesting issue
for further research.

\

\section{Evidence: who are the gossips?}

We now move to the second of our main results.

\subsection{Data Collection}\label{sec:data}

\

As an empirical study of people's ability to nominate central individuals, we use a rich
network data set that we gathered from villages in rural Karnataka (India).
We collected network data in 2006 in order to study the spread of their microfinance product \citep{banerjeecdj2013}. We again collected network data in 2012, which is the data we use here.

We use the the network data combined
with ``gossip'' information from 33 villages. To collect the network data (described in detail in \citet*{banerjeecdj2013},
and publicly available at \url{http://economics.mit.edu/faculty/eduflo/social}),
we asked adults to name those with whom they interact in the course
of daily activities.\footnote{We have network data from 89.14\% of the 16,476 households based on
interviews with 65\% of all adult individuals aged 18 to 55.  This is a new wave of data relative to our
original microfinance study.}  We have data concerning 12 types of interactions for a given survey
respondent: (1) whose houses he or she visits, (2) who visits his
or her house, (3) his or her relatives in the village, (4) non-relatives
who socialize with him or her, (5) who gives him or her medical help,
(6) from whom he or she borrows money, (7) to whom he or she lends
money, (8) from whom he or she borrows material goods (e.g., kerosene,
rice), (9) to whom he or she lends material goods, (10) from whom
he or she gets important advice, (11) to whom he or she gives advice, and
(12) with whom he or she goes to pray (e.g., at a temple, church or
mosque). Using these data, we construct one network for each village,
at the household level, where a link exists between households if any
member of either household is linked to any other member of the other
household in at least one of the 12 ways. Individuals can communicate
if they interact in any of the 12 ways, so this is the network of
potential communications, and using this network avoids the selection
bias associated with data-mining to find the most predictive subnetworks.
The resulting objects are undirected, unweighted networks at the household
level.

After the network data were collected, to collect gossip data, we asked the
adults the following two additional questions:
\begin{itemize}
\item[(Event)] \emph{If we want to spread information to everyone in the village
about tickets to a music event, drama, or fair that we would like
to organize in your village, to whom should we speak?}

\

\item[(Loan)] \emph{If we want to spread information about a new loan product to
everyone in your village, to whom do you suggest we speak?}
\end{itemize}

We asked two questions to check whether there was
any difference between depending on what the people thought was being diffused.  It made no difference, as is clear from the results below.

Table \ref{tab:summary} provides summary statistics.
The networks are sparse: the average number of households in
a village is 196 with a standard deviation of 61.7, while the average
degree is 17.7 with a standard deviation of 9.8.

Only half of the households were willing to respond to our  ``gossip'' questions.
This is in itself intriguing. Perhaps people are unwilling to offer
an opinion when they are unsure of the answer.\footnote{See \citet{alatas2012network} for a model that builds on this idea.}
They might instead have been worried about singling someone out.

Conditional on naming someone, however, there is substantial concordance
of opinion. Only 4\% of households were nominated in response to the
Event question (and 5\% for the Loan question) with a cross-village
standard deviation of 2\%. Conditional on being nominated, the median
household was nominated nine times.\footnote{We work at the household level, in keeping with \citet{banerjeecdj2013}
who used households as network nodes; a household receives a nomination
if any of its members are nominated.} This is perhaps a first indication that the answers may be meaningful,
since if people are good at identifying central individuals, we would
expect their nominations to coincide.

In this data set, we label as ``leaders'' households that contain shopkeepers, teachers, and leaders of self-help
groups -- almost 12 percent
of households fall into this category. This was how the bank in our microfinance study defined
leaders, who were identified as people to be seeded with information
about their product (because it was believed they would be good at transmitting the
information).  The bank's theory was that such leaders were
 likely to be well-connected in the villages and
thereby would contribute to more diffusion of microfinance.\footnote{In our earlier work, \citet{banerjeecdj2013}, we show that there
is considerable variation in the centrality of these ``leaders''
in a network sense, and that this variation predicts the eventual
take up of microfinance.}

There is some overlap between leaders and gossip nominees. We refer to the nominees as ``gossips.''  Overall, 86\% of the population were neither gossips nor leaders, just 1\% were both, 3\% were nominated but not leaders, and 11\% were leaders but not nominated. This means that 8\% of leaders were nominated as a gossip under the event question whereas 92\% where not nominated. Similarly, 25\% of nominated gossips under the event question were leaders, whereas 75\% were not. The loan question demonstrates very similar results, and Figure \ref{fig:cdfCentrality} presents this information.

\subsection{Do individuals nominate central nodes?}
\label{sec:results}

\

Our theoretical results suggest that people can learn others' diffusion
centralities simply by tracking news that they hear
through the network, and therefore should be able to name central individuals
when asked whom to use as a ``seed'' for diffusion. In this section,
we examine whether this is the case.

\subsubsection{Data description}

\

As motivating evidence, Figure \ref{fig:cdfCentrality} shows the distribution of diffusion centrality
(normalized by its standard deviation across the sample for interpretability)
across households that were nominated for
 the event question, those who were nominated as leaders, and those who were named for both or neither. Very clearly,  the distribution of centrality of those who are both nominated and are also leaders first order stochastically dominates the other distributions. Moreover, the distribution of centralities of those who are nominated but not leaders dominates the distribution of those who are leaders but  were not nominated. Finally, those who are neither nominated nor a leader exhibit a distribution that is dominated by the rest. Taken together, this shows that individuals who are both nominated and leaders tend to be more central than those who are nominated but not leaders, who are in turn more central than those who are not nominated but are leaders.

Figure \ref{fig:neighborhood} presents the distribution of nominations
as a function of the network distance from a given household. If information
did not travel well through the social network, we might imagine that
individuals would only nominate households with whom they are directly
connected. Panel A of Figure \ref{fig:neighborhood} shows that fewer
than 13\% of individuals nominate someone within their direct neighborhood,
compared to about 9\% of nodes within this category. At the same time,
over 28\% of nominations come from a network distance of at least
three or more (41\% of nodes are in this category). Therefore, although respondents tend to nominate people who are closer to them than the average
person in the village, they are also quite likely to nominate someone
who is far away. Moreover, it is important to note that highly central
individuals are generally closer to people than the typical household
(since they have many friends -- the famous ``friendship paradox''),
so it does make sense that people tend to nominate individuals who
are closer to them. Taken together, this suggests that information
about centrality does indeed travel through the network.

Panel B of Figure \ref{fig:neighborhood} shows that
the average diffusion centrality in percentile terms of those named at distance
1 is higher than of those named at distance 2, which is higher than of those named at distance 3 or more.
This suggests that individuals have more accurate information about central individuals that are closer to them, and when they don't, they are careful not to nominate (recall that fewer than half of the households nominate anyone).

\subsubsection{Regression Analysis}
\label{sec:prediction}

\

Motivated by this evidence, we present a more systematic analysis
of the correlates of nominations, using a discrete choice framework
for the decision to nominate someone.

Our theory suggests that if people choose whom to nominate based on
who they hear about most frequently, then diffusion centrality should
be a leading predictor of nominations. While the aforementioned results
are consistent with this prediction, there are several plausible alternative
interpretations that do not rely on the information mechanism from our model.
For example, individuals may nominate the person
who has the most friends, and people with many friends tend to be more
diffusion central than those with fewer friends (i.e., diffusion centrality
and degree centrality are correlated). Alternatively,
it may be that people simply nominate the ``leaders'' within their
village, or people who are central geographically, and these also
correlate with diffusion/eigenvector centrality. There are reasons to think that leadership status and geography may
be good predictors of network centrality, since, as noted in \citet{banerjeecdj2013},
the microfinance organization selected ``leaders'' precisely because
they believed these people would be central. Previous
research has also shown that geographic proximity increases the probability
of link formation \citep{fafchampsg2007b,ambrusms2012,chandrasekharl2014}
and therefore, one might expect geographic data to be a useful predictor
of centrality. For that reason, since in addition to leadership data we have detailed  GPS
coordinates for every household in each village, we include these
in our analysis below as controls.\footnote{To operationalize geographic centrality, we use two measures. The
first uses the center of mass. We compute the center of mass and then
compute the geographic distance for each agent $i$ from the center
of mass. Centrality is the inverse of this distance, which we normalize
by the standard deviation of this measure by village. The second uses
the geographic data to construct an adjacency matrix. We denote the
$ij$ entry of this matrix to be $\frac{1}{d(i,j)}$ where $d(\cdot,\cdot)$
is the geographic distance. Given this weighted graph, we compute
the eigenvector centrality measure associated with this network. Results
are robust to either definition.}

We recognize that the correlations below do not constitute
proof that the causal mechanism is indeed gossip, but they do rule
out these obvious confounding factors.

To operationalize our analysis
we use $DC\left(1/\E[\lambda_{1}],\E[Diam(\g(n,p))]\right)$ as our
measure of diffusion centrality, as discussed in Section \ref{sec:properties}.

We estimate a discrete choice model of the decision to nominate an
individual. Note that we have large choice sets, as there are $n-1$
possible nominees and $n$ nominators per village network. We model
agent $i$ as receiving utility $u_{i}(j)$ for nominating individual
$j$:
\[
u_{i}(j)=\alpha+\beta'x_{j}+\gamma'z_{j}+\mu_{v}+\epsilon_{ijv},
\]
where $x_{j}$ is a vector of network centralities for $j$ (eigenvector
centrality, diffusion centrality, and degree centrality), $z_{j}$
is a vector of demographic characteristics (e.g., leadership status,
geographic position, and caste controls), $\mu_{v}$ is a village fixed
effect, and $\epsilon_{ijv}$ is a Type-I extreme value distributed
disturbance. For convenience given the large choice sets, we estimate
the conditional logit model by an equivalent Poisson regression, where
the outcome is the expected number of times an alternative is selected
\citep{palmgren1981fisher,baker1994multinomial,lang1996comparison,guimaraes2003tractable}. This is presented in Table \ref{tab:predict1}.
A parallel OLS specification leads to the same conclusion, and is presented in Appendix \ref{sec:phase1_ext}.

We begin with a number of bivariate regressions in Table \ref{tab:predict_biv}.
First, we show that diffusion centrality is a significant driver of
an individual nominating another (column 1). A one standard deviation
increase in diffusion centrality is associated with a 0.607 log-point
increase in the number of others nominating a household (statistically
significant at the 1\% level). Columns 2 to 5 repeat the exercise
with two other network statistics (degree and eigenvector centrality),
with the ``leader'' dummy, and with an indicator for geographic centrality.
All of these variables, except for geographic centrality, predict
nomination, and the coefficients are similar in magnitude.

The different network centrality measures are all correlated. To investigate
whether diffusion centrality remains a predictor of gossip nomination
after controlling for the other measures, we start by introducing
them one by one as controls in column 1 to 4 in Table \ref{tab:predict_biv}.
Degree is insignificant, and does not affect
the coefficient of diffusion centrality. Eigenvector centrality is quite correlated
with diffusion centrality (as it should be, since they converge to each other with
enough time periods), and hard to distinguish from it. Introducing it cuts the effect of diffusion centrality by about
50\%, though it remains significant. The leader dummy is close to being
significant, but the coefficient of diffusion centrality remains
strong and significant. The geographic centrality variable now has
a negative coefficient, and does not affect coefficient of the diffusion
centrality variable.

These results provide suggestive evidence that a key driver of the nomination
decision involves diffusion centrality with $T>1$, although it may be more difficult
to separate eigenvector centrality and diffusion centrality from each other, which is not surprising since they are closely related
concepts.

To confirm this pattern,
in  the last column, we introduce all the variables together and perform a LASSO analysis, which
``picks'' out the variable that is most strongly associated with the outcome variable, the number of nominations. Specifically, we use the post-LASSO procedure of \cite{belloni2009least}. It is a two-step procedure. In the first step, standard LASSO is used to select the support: which variables matter in predicting our outcome variable (the number of nominations). In the second step, a standard Poisson regression is run on the support selected in the first stage.\footnote{The problem with the returned coefficients from LASSO in the first step is that it shrinks the coefficients towards zero. \cite{belloni2009least}, \cite{belloni2014inference} and \cite{belloni2014high} show that running the usual OLS (in our case, Poisson) on the variables selected in the first stage in a second step will recover consistent estimates for the parameters of interest.}${}^,$\footnote{To our knowledge, the post-LASSO procedure has not been developed for nonlinear models, so we only conduct the selection using OLS.}

As we did before, we consider the variables diffusion centrality, degree centrality, eigenvector centrality, leadership status, and geographic centrality in the standard LASSO to select the support. For the event nomination, LASSO picks out only one predictor: diffusion centrality. The post-LASSO
coefficient and standard error thus exactly replicate the Poisson regression of using just DC(0.2,3). This confirms that diffusion centrality is the key
predictor of gossip nomination. For the loan nomination, the LASSO picks out both degree and diffusion as relevant, though degree is insignificant.
We repeat the analysis with OLS instead of Poisson regression in Appendix \ref{sec:phase1_ext}, with identical qualitative results.

\section{Experiment: Do gossip nominees spread information widely?}\label{sec:experiment}

We have shown that individuals nominate diffusion central people.

Our third main investigation is to directly test our method.   Using our gossip nomination
protocol, do we get better diffusion
compared to other ways of choosing the seeds?
This is a key policy implication of our theory.

\subsection{Information Diffusion and Gossip Seeding}

\

We compare seeding of information to gossips (nominees) to two
benchmarks: (1) a set of village elders and (2) randomly selected
households. Seeding information among random households provides the
most relevant benchmark, because it allows us to study how information
circulates starting from random households. Seeding information
with village elders provides an interesting benchmark, because they
are traditionally respected as social and political leaders and one might presume that they are the right place to start. They have the advantage of being easy to identify, and it could be, for instance, that
information spreads widely only if it has the backing of someone who
can influence opinion, not just convey information.

We conducted an experiment in 213 villages in Karnataka that were
not involved in the microfinance diffiusion project and that were not villages in which we had previously worked. In every village, we attempted to contact $k$ households
and inform them about a promotion run by our partner, a cellphone
sales firm.  The promotion gave villagers a non-rivalrous chance
to win a new mobile phone or a cash prize.

The promotion worked as follows. Anyone who wanted to participate could give us
a ``missed call'' (a call that we registered, but did not answer, and which was thus free).
In public, a few weeks later, the registered phone numbers were randomly awarded
cash prizes ranging from 50 to 275 rupees, or a free cell phone.  Which prize any given entrant was
awarded was determined by the roll of two dice, ensuring that the
awarding of prizes was non-rivalrous and there was no strategic incentive to withhold information about the promotion.

In each treatment, the seeded individuals were encouraged to inform others in their community
about the promotion.  Our primary outcome data is thus the number of calls from unique households that
we received.\footnote{The calls from the seeds are included in the main specification, and so we include seed number fixed effects.}
In half of the villages, we set $k=3$, and in half of the villages
we set $k=5$. This was done because we were not sure of the right number of seeds that were needed to avoid either the process dying out or complete
and rapid diffusion. In practice, we find that there is no significant difference between 3 and 5 seeds on the outcome variable (number of calls received).

We randomly divided the sample of 213 into three arms of 71 villages,
where the $k$ seeds were selected as follows.
A few days before the experiment, we interviewed up to 15 households
in every village (selected randomly via circular random sampling via
the right-hand rule method) to identify ``elders'' and ``gossips,'' as described below.\footnote{Circular sampling is a standard survey methodology where the enumerator
starts at the end of a village, and, using a right-hand rule, spirals
throughout the entire village, when enumerating households. This allows us to cover the entire geographic span of the village which is desirable in this application, particularly as castes are often segregated, which may lead to geographic segregation of the network: we want to make sure the nominations reflect the entire village.}$^,$
\footnote{We asked the same questions in all villages so to be sure that the surveying had no impact on the treatment, and to allow
us to rack which sorts of seeds were reached in each treatment.  }
\begin{itemize}
\item[T1.] Elder: $k$ households were chosen from the list of village
elders obtained one week prior. The notion of ``village elder'' is well recognized in these villages: there are people who are recognized authorities, and believed
to be influential.
\item[T2.] Random: $k$ households were chosen uniformly at random, also using
the right-hand rule method and going to every $n/k$ households.
\item[T3.] Gossip: $k$ households were chosen from the list of gossip nominees
obtained one week prior.

\end{itemize}

Note that this seeding does not address the challenging problem of choosing the optimal set of nodes for diffusion given their centralities.  The solution is not simply to pick
the highest ranked nodes, since the positions of the seeds relative to each other also matters.  This results in a computationally challenging problem (in fact, an NP-complete one, see \citet*{kempekt2003,kempekt2005}).  Here, we  randomly selected
seeds from the set of nominees, which if anything biases the test against the gossip treatment. We could have instead used the most highly nominated nodes in combination with caste or other demographic information to pick combinations of highly central nodes that are likely to be well-spaced in the network.

The main outcome variable that we are interested in is the number of calls
received.
This represents the number of people who heard about
the promotion and wanted to participate. The mean number of calls in the sample is 9.35 (with standard deviation 15.64). The median number of villagers who participated is 3 across all villages.
In 80.28\% of villages, we received at least one call, and the 95th percentile is 39. This is not a very large number, which is not surprising, since many villagers already own cell phones and may not have been particularly interested in a chance of winning another one. Nonetheless, there is enough variation from village to village to allow us to identify the effect of information diffusion. Note that, given these small numbers, we exclude one village in our analysis in which the number of calls was 106. In this village one of  the seeds (who happened to be a gossip nominee in a random village) prepared posters to broadcast the information broadly. The diffusion in this villages does not have much to do with the network process we have in mind.
We thus use data from 212 villages in all the regressions that follow.  The results including this village are presented in Appendix \ref{sec:broadcast}. They
are qualitatively similar, but the OLS of the impact of hitting at least one gossip is larger and more precise, while the Reduced from and IV estimates are similar but noisier.

Figure \ref{fig:experiment} presents the results graphically. The distribution of calls in the gossip villages clearly stochastically
dominates that of the elder and random graphs. Moreover, the incidence
of a diffusive event, where a large number of calls is received,
is rare when we seed information randomly or with village elders --
but we do see such events when we seed information with gossip
nominees.

We begin with the analysis of our experiment as designed, which is the policymaker's experiment: what is the impact on diffusion of purposefully seeding gossips or elders,
as compared to random villagers?

\begin{align}\label{eq:RF}
y_{j} = \ &\theta_0 + \theta_1 GossipTreatment_j + \theta_2 ElderTreatment_j + \\
& \theta_3 NumberSeeds_j + \theta_4 NumberGossip_j + \theta_5 NumberElder_j + u_j, \nonumber
\end{align}
where $y_j$ is the number of calls received from village $j$ (or the number of calls per seed),
$GossipTreatment_j$ is a dummy equal to 1 if seeds were assigned to be from the gossip list,
$ElderTreatment_j$ is a dummy equal to 1 if seeds were assigned to be from the elder list,
$NumberSeeds_j$ is the total number of seeds, 3 or 5, in the village,
$NumberGossip_j$ is the total number of gossips nominated in the village, and
$NumberElder_j$ is the total number of elders nominated in the village.

Table \ref{tab:mainExpt_2_1} presents the regression analysis. The results including the broadcast village are presented in Appendix \ref{sec:broadcast}.\footnote{The OLS specification is larger, while the IV has a similar point estimate but is noisier.} Column 1 shows the
reduced form (\ref{eq:RF}). In control villages, we received 8.077 calls, or an average of 1.967 per seeds. In gossip treatment villages, we received 3.65 more calls ($p = 0.19$)
in total or 1.05 per seed ($p = 0.13$).

While this exercise is of independent interest, it is not the most direct test of our theory.
Our theory predicts that information seeded with a gossip will
flow faster than information seeded to someone who is not a gossip. Thus,
the most direct test of our model is to compare diffusion in villages
in which a gossip was hit to diffusion in villages where no gossip was hit.
The  seeding does not exclude gossips in the random and elder treatment villages.
In some random and elder treatment villages, gossip nominees  were included in our seeding
set by chance. On an average, 0.59 seeds were gossips in random villages.

Our next specification is thus to compare villages where ``at least 1 gossip was hit,'' or  ``at least 1 elder was
hit'' (both could be true simultaneously) to those where no elder or no gossip was hit (the control group).
Although the selection of households under treatments is random, the event that at least one gossip (elder) being hit is random only conditional on the number of potential gossip (elder)
seeds present in the village. We thus include as controls in the OLS regression of
number of calls on ``at least 1 gossip (elder) seed hit''. This specification should give us the causal effect of gossip (elder) seeding, but to assess the robustness,
we also make directly use of the variation induced by the village level experiment, and we instrument ``at least 1 gossip (elder) seed hit'' is instrumented by the gossip (elder) treatment status of the village.

Therefore, we are interested in

\begin{align}\label{eq:OLSIV}
y_{j} = \ & \beta_0 + \beta_1 GossipReached_j + \beta_2 ElderReached_j + \\
&\beta_3 NumberSeeds_j + \beta_4 NumberGossip_j + \beta_5 NumberElder_j + \epsilon_j. \nonumber
\end{align}

\noindent

This equation is estimated either by OLS, or by instrumental variables, instrumenting $GossipReached_j$ with $GossipTreatment_j$ and $ElderReached_j$ with $ElderTreatment_j$. There the first stage equations are

\begin{align}\label{eq:FS-gossip}
GossipReached_j = \ &\pi_0+ \pi_1 GossipTreatment_j + \pi_2 ElderTreatment_j + \\ \nonumber
& \pi_3 NumberSeeds_j + \pi_4 NumberGossip_j + \pi_5 NumberElder_j + v_j, \\ \nonumber
\end{align}
and

\begin{align}\label{eq:FS-elder}
ElderReached_j = \ &\rho_0+ \rho_1 GossipTreatment_j + \rho_2 ElderTreatment_j + \\ \nonumber
& \rho_3 NumberSeeds_j + \rho_4 NumberGossip_j + \rho_5 NumberElder_j + \nu_j. \\ \nonumber
\end{align}

Column 2 of Table \ref{tab:mainExpt_2_1}  shows the OLS. The effect of hitting at least one gossip seed is 3.79 for the total number of calls (p-value = 0.04) ,which represents a 65\% increase, relative to villages where no gossip seed was hit,
or  0.95 (p-value = 0.06)
calls per seed. Column 5 presents the IV estimates (Columns 3 and 4 present the first stage results for the IV). They are larger than the OLS estimates, and statistically indistinguishable, albeit less precise.

Given the distribution of calls, the results are potentially sensitive to outliers. We therefore present quantile regressions of the comparison between gossip/no gossip and Gossip treatment/Random villages in Figure \ref{fig:Quantilev2}. The specification that compares villages with or without gossip hit (Panel B) is much more precise. The treatment effects are significantly greater than zero starting at the 35th percentile.
Specifically, hitting a gossip significantly increases the median number of calls by 122\% and calls at the 80th percentile by 71.27\%.

This is our key experimental result: gossip nominees are much better
seeds for diffusing a piece of information, at least in these experiments.

\subsection{Mechanism: does gossip seed diffusion capture diffusion centrality?}

\

We have seen, in the first part of our empirical investigation, that villagers
nominate individuals who are diffusion central.  To what extent is the
greater diffusion of information in the experiment mediated by the
diffusion centrality of the gossip seeds, and to what extent does
it reflect the villagers' ability to capture other dimensions of individuals
that would make them good at diffusing information?

To get at this issue, a few weeks after the experiment, we collected
network data in 69 villages in which seeds were randomly
selected (2 of the 71 villages were not accessible at the time). In these
villages, by chance, some seeds happened to be gossips and/or elders.
We create a measure of centrality that parallels the gossip
dummy and elder dummy by forming a dummy for ``high diffusion
centrality.'' We defined a household has ``high diffusion centrality''
if its diffusion centrality is at least one standard deviation above
the mean. With these measures, in our 69 villages, 13\% of households are defined to have
``high diffusion centrality'', while 1.7\% were nominated as seeds, and 9.6\% are ``leaders.''
Twenty-four villages have exactly one high diffusion centrality seed and
14 have more than one. Twenty-three villages have exactly one
gossip seed, and 8 have more than one.\footnote{We continue to exclude the one village
in which a gossip seed broadcasted information. The results including that village are in Appendix \ref{tab:seedType_OLS_flyer}: they
reinforce the conclusion that diffusion centrality does not capture everything about why gossips are good seeds, since this particular
gossip seed had low diffusion centrality. With this village in, the coefficient of hitting at least one gossip does not decline when we control for diffusion centrality,
and in fact diffusion centrality, even on its won, is not significantly associated with more diffusion.}

Column 1 of Table \ref{tab:seedType_OLS} runs the same specification as in Table \ref{tab:mainExpt_2_1} but in the
 68 random villages. In these villages, hitting a gossip by chance increases the number of calls by 6.65 (compared to 3.78 in the whole sample).
In column 3, we regress the number of calls on a dummy for hitting = a high $DC$ seed: high $DC$ seeds
do increase the number of calls (by 5.18 calls).
Finally, we regress  number of calls received only on dummy of hitting a high $DC$ seed, and we see that the number of calls increase by 5.18.
In column 2, we augment the specification in column 1 to add the dummy for ``at least one DC central seed''. Since DC and Gossip are correlated,
the regression is not particularly precise. The point estimate of gossip, however, only declines slightly.

Taking the point estimates seriously, we see that the results suggest that diffusion centrality captures
part of the impact of a gossip nomination, but likely not all of it.
Gossip seeds
tend to be highly central, and information does spread more from
highly central seeds. This accounts for some part of the reason why information
diffuses more extensively from gossip nominated seeds.  At the same time, it
is also appears that the model does not capture the entire reason why
gossip seeds are best for diffusing information: even controlling for
their diffusion centrality, gossip seeds still lead to greater diffusion.
It is likely that our measures of the network are imperfect, and so part of the extra diffusion from the gossip nominations
could reflect that villagers have better estimates of diffusion centrality from their network gossip than we do from our surveys.
It also could be
that the gossip nomination is a richer proxy
for information diffusion than the model-based centrality measure.
For instance, there are clearly other factors that
predict whether a seed will be good at diffusing information beyond
their centrality (altruism, interest in the information, etc.) and
villagers may be good at capturing those factors. However, the
standard errors do not allow us to pinpoint how much of the extra diffusion
coming from being nominated
as a gossip is explained by network centrality.

\section{Conclusion}\label{sec:conclusion}

Our model illustrates that it should be easy for even very myopic
and non-Bayesian (as well as fully rational) agents, simply by counting, to have an idea of who
is central in their community  -- according to fairly complex measures of centrality.
Motivated by this, we asked villagers to identify good diffusers
in their village. They do not simply name locally central individuals
(the most central among those they know), but actually name people who
are {globally} central within the village. Moreover, in a specially
designed experiment, we find that nominated individuals are indeed
much more effective at diffusing a simple piece of information than other individuals, even village elders.
This suggests that people can
use simple observations to learn valuable things about the complex
social systems within which they are embedded, and that researchers and others who are interested
in diffusing information have an easy and direct method of identifying highly central seeds.

Although our model focuses on the network-based mechanics of communication,
in practice, considerations beyond simple network position may determine
who the ``best'' person is to spread information, as other characteristics
may affect the quality and impact of communication. It seems that
villagers take such characteristics into account and thus nominate
individuals who are not only highly central but who are even more successful at
diffusing information
than the average highly central individual.

Our findings have important policy implications, since such nominations
are easy to collect and therefore can be used in a variety of
contexts, either on their own or combined with other easily collected
data, to identify effective seeds for information diffusion. Thus,
using this sort of protocol may be a cost-effective way to improve
diffusion and outreach.

Beyond these applications, the work presented here opens a rich agenda for further
research, as one can explore which other aspects of agents' social
environments can be learned in simple ways. For example, can individuals
also identify individuals who are trusted by others? A
piece of information about a cell-phone giveaway is probably innocuous enough to
be transmitted by a ``gossip,'' but what about advice on immunization,
for example?

\bibliographystyle{ecta}
\bibliography{networks}

\newpage

\appendix

\section*{Figures}

\begin{figure}[!h]
\begin{center}
\subfloat[Event question]{
\includegraphics[trim = 2.5mm 4.7mm 2.5mm 2.5mm, clip = true, scale = 0.65]{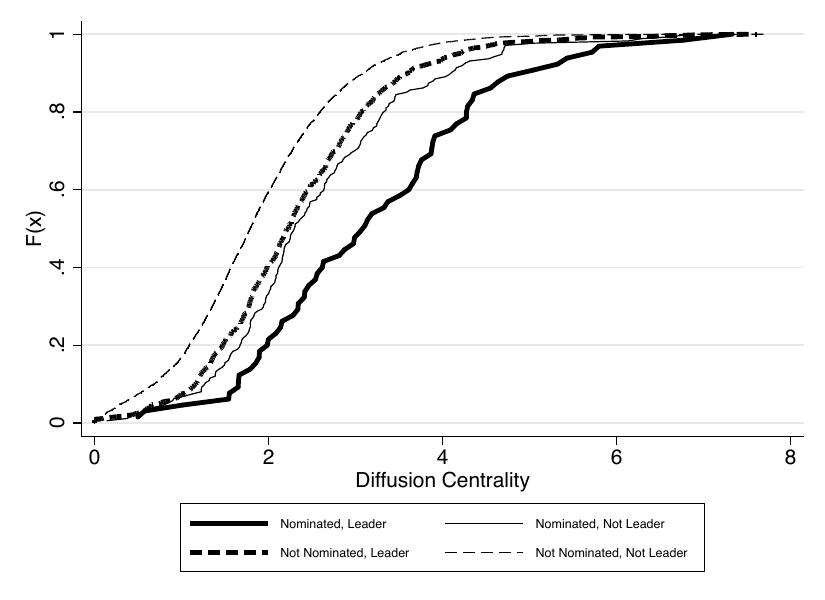} 
}

\

\scalebox{0.6}{\begin{tabular}{l*{1}{c}}
                                                                                                                                                                                                        &population share\\
nominated, leader (event)                                                                                                                                                                               &        0.01\\
not nominated, leader (event)                                                                                                                                                                           &        0.11\\
nominated, not leader (event)                                                                                                                                                                           &        0.03\\
not nominated, not leader (event)                                                                                                                                                                       &        0.86\\
\end{tabular}
}

\

\subfloat[Loan question]{
\includegraphics[trim = 2.5mm 4.7mm 2.5mm 2.5mm, clip = true, scale = 0.65]{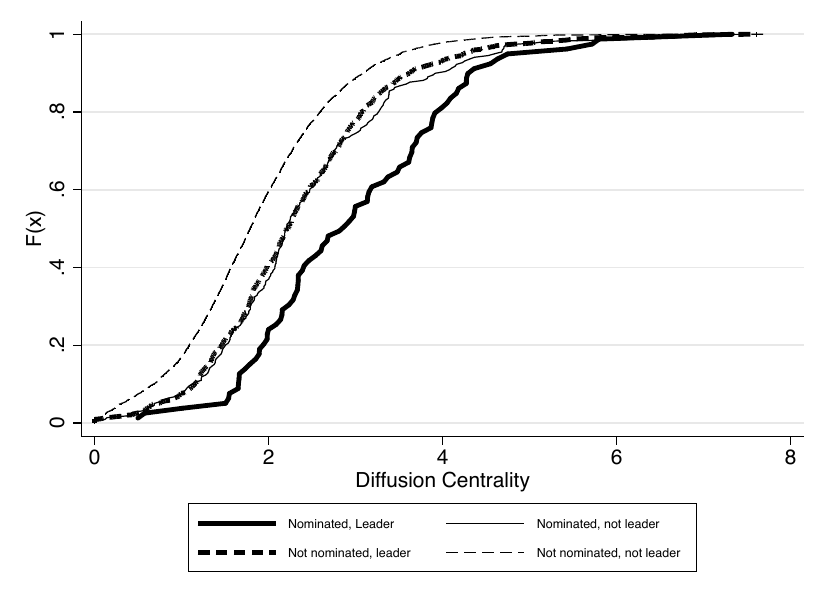}
}
\

\scalebox{0.6}{\begin{tabular}{l*{1}{c}}
                                                                                                                                                                                                        &population share\\
nominated, leader (loan)                                                                                                                                                                                &        0.01\\
not nominated, leader (loan)                                                                                                                                                                            &        0.10\\
nominated, not leader (loan)                                                                                                                                                                            &        0.03\\
not nominated, not leader (loan)                                                                                                                                                                        &        0.85\\
\end{tabular}
}

\end{center}
\caption{This figure uses the Phase 1 dataset. It presents CDFs of the (normalized) diffusion centrality, diffusion centrality divided by the standard deviation, conditional on  classification (whether or not it is nominated under the event question in Panel A and the loan question in Panel B and whether or not it has a village leader).} \label{fig:cdfCentrality}
\end{figure}

\begin{figure}[!h]
\begin{center}
	\subfloat[Share of nominees in specified neighborhood]{
	\includegraphics[trim = 2.5mm 4.7mm 2.5mm 2.5mm, clip = true, scale = 0.65]{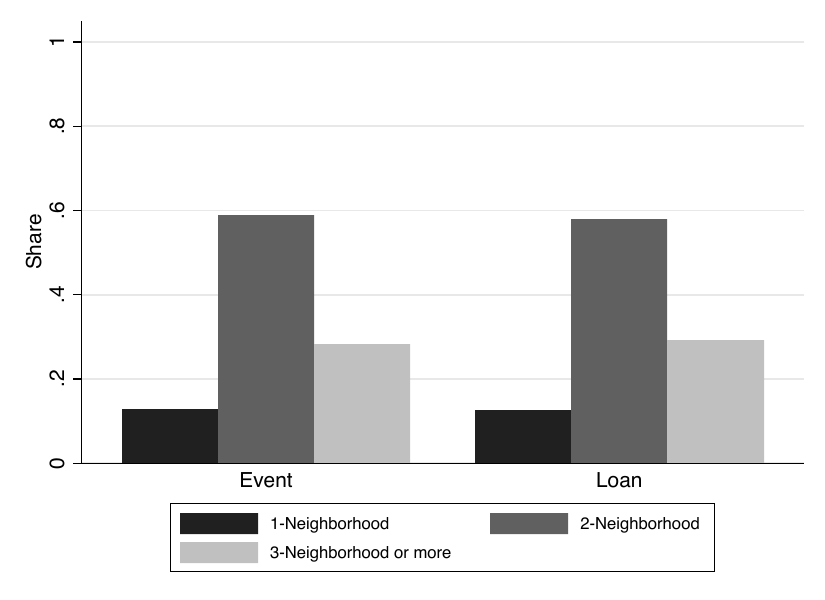} 
	}

	\
	\scalebox{0.6}{\begin{tabular}{@{\hskip\tabcolsep\extracolsep\fill}l*{1}{c}}
                                                                                                    &       share\\
Nodes in 1-Neighborhood                                                                             &        0.09\\
Nodes in 2-Neighborhood                                                                             &        0.50\\
Nodes in 3-Neighborhood or more                                                                     &        0.41\\
\end{tabular}
}
	\

	\subfloat[Average diffusion centrality percentile of nominees in specified neighborhood]{
	\includegraphics[trim = 2.5mm 4.7mm 2.5mm 2.5mm, clip = true, scale = 0.65]{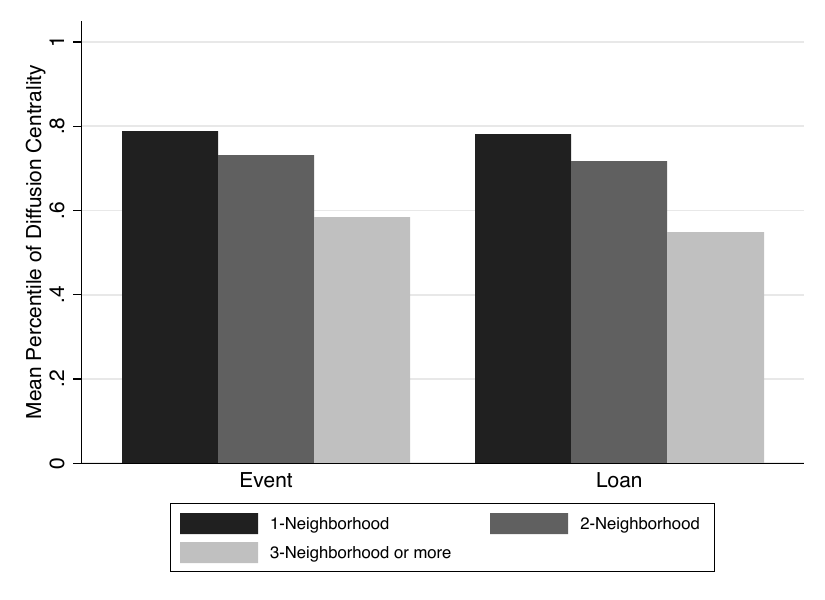}
	}
\end{center}
		\caption{Distribution of nominees and their diffusion centrality by network distance in the Phase 1 dataset.} \label{fig:neighborhood}
\end{figure}

\clearpage

\begin{figure}[!h]
\begin{center}
	\includegraphics[trim = 2.5mm 4.7mm 2.5mm 2.5mm, clip = true, scale = 1]{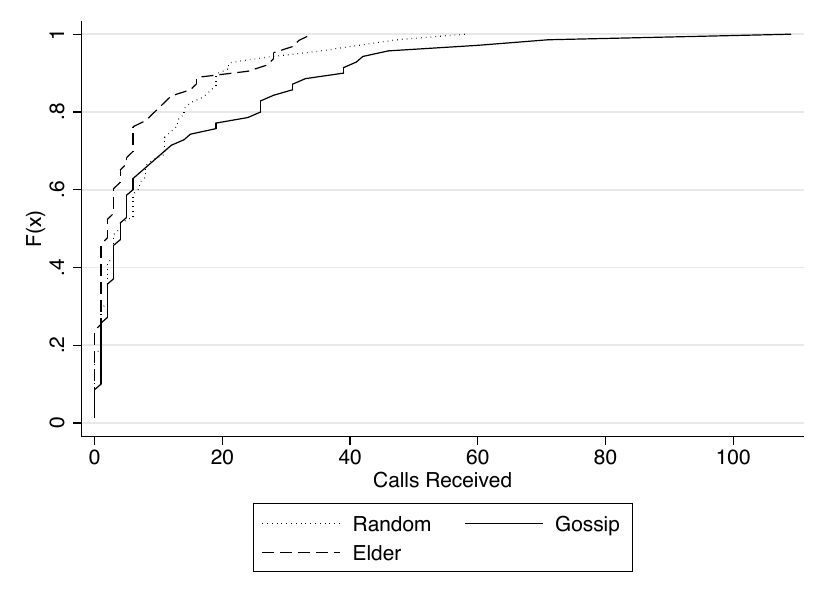}
\end{center}
		\caption{Distribution of calls received by treatment in the Phase 2 experiment.\label{fig:experiment}}
\end{figure}

\clearpage

\begin{figure}[!h]
\begin{center}
	\subfloat [\small{Quantile treatment effect  by treatment - Reduced Form}]{
	\includegraphics[scale = 0.9]{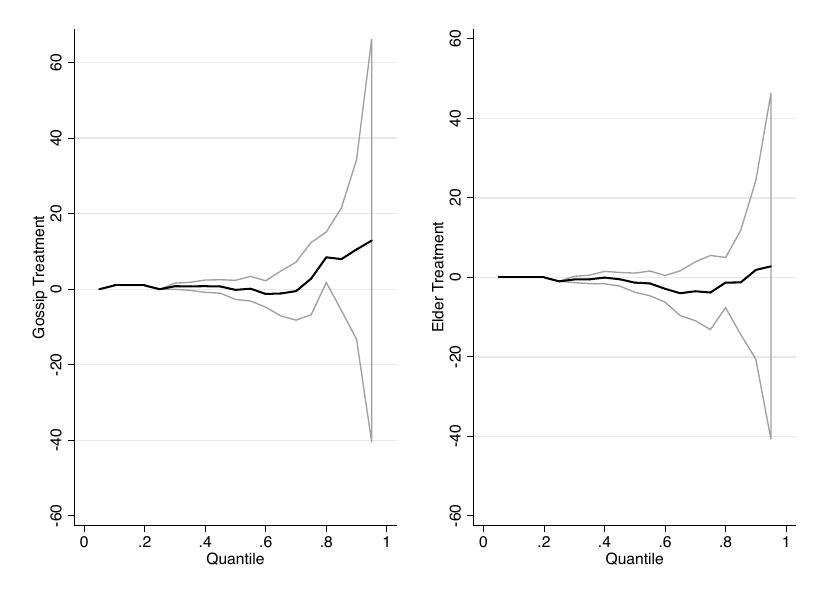}
	}
	
	\
	
	\subfloat[\small{Quantile treatment effect by hitting at least one gossip or elder}]{
	\includegraphics[scale = 0.9]{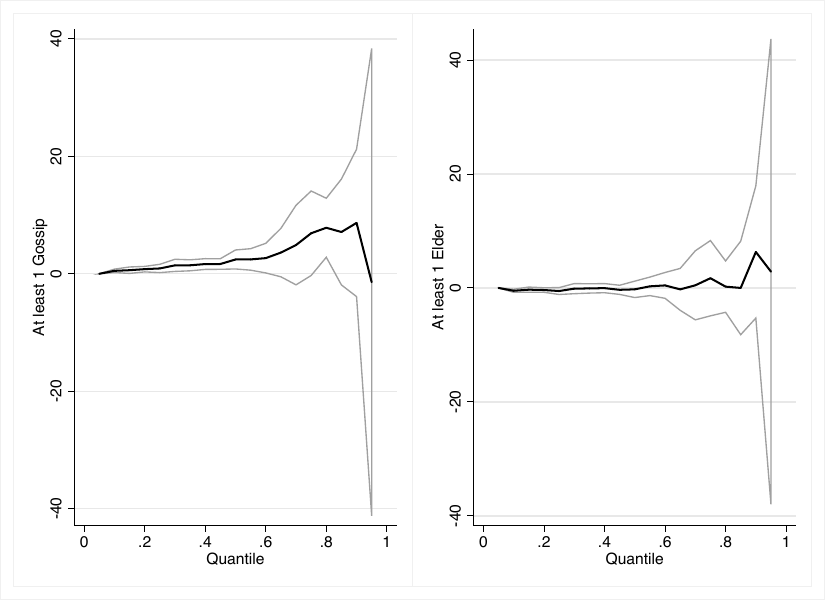}
	}
	\end{center}
		\caption{Quantile treatment effects where for $j \in \{Gossip, \ Elder\}$,  $\widehat{\beta}_j(u)$  is computed for $u = \{0.05,...,0.95\}$. The intercept $\alpha(u)$ (not pictured) in each case is the omitted category corresponding to the random treatment.\label{fig:Quantilev2}}
\end{figure}

\clearpage

\section*{Tables}

\begin{table}[!h]
\centering
\caption{Summary Statistics\label{tab:summary}}
\begin{threeparttable}
{
\def\sym#1{\ifmmode^{#1}\else\(^{#1}\)\fi}
\begin{tabular}{l*{1}{cc}}
\hline\hline
                                                                                                                                                                                                        &        mean&          sd\\
\hline
households per village                                                                                                                                                                                  &      196&       61.70\\
household degree                                                                                                                                                                                        &       17.72&        9.81\\
clustering in a household's neighborhood                                                                                                                                                                &        0.29&        0.16\\
avg distnace between nodes in a village                                                                                                                                                                 &        2.37&        0.33\\
fraction in the giant component                                                                                                                                                                         &        0.98&        0.01\\
is a leader                                                                                                                                                                                             &        0.12&        0.32\\
nominated someone for event                                                                                                                                                                             &        0.38&        0.16\\
nominated someone for loan                                                                                                                                                                              &        0.48&        0.16\\
was nominated for event                                                                                                                                                                                 &        0.04&        0.2\\
was nominated for loan                                                                                                                                                                                  &        0.05&        0.3\\
number of nominations received for event                                                                                                                                                                &        0.34&        3.28\\
number of nominations received for loan                                                                                                                                                                 &        0.45&        3.91\\
\hline\hline
\end{tabular}
}

\begin{tablenotes}
Notes: This table presents summary statistics from the Phase 1 dataset: 33 villages of the Banerjee et al. (2013) networks dataset where nomination data was originally collected in 2011/2012. For the variables ``nominated someone for loan (event)'' and ``was nominated for loan (event)'' we present the cross-village standard deviation.
\end{tablenotes}
\end{threeparttable}
\end{table}
\ 

\begin{table}[!h]
\centering
\caption{Leader Gossip Overlap} \label{tab:overlap}
\begin{threeparttable}
{
\def\sym#1{\ifmmode^{#1}\else\(^{#1}\)\fi}
\begin{tabular}{@{\hskip\tabcolsep\extracolsep\fill}l*{1}{c}}
\hline\hline
                                                                                                    &       share\\
\hline
leaders who are nominated (loan)                                                                    &        0.11\\
nominated who are leaders (loan)                                                                    &        0.27\\
leaders who are not nominated (loan)                                                                &        0.89\\
nominated who are not leaders (loan)                                                                &        0.73\\
leaders who are nominated (event)                                                                   &        0.09\\
nominated who are leaders (event)                                                                   &        0.27\\
leaders who are not nominated (event)                                                               &        0.91\\
nominated who are not leaders (event)                                                               &        0.73\\
\hline\hline
\end{tabular}
}
 
\begin{tablenotes}
Notes: This table presents the overlap between ``leaders'' in the sample and those nominated as gossips (for loan and event). 
\end{tablenotes}
\end{threeparttable}
\end{table}

\clearpage

\begin{table}[!h]
\centering
\caption{Factors predicting nominations}\label{tab:predict1}
\begin{threeparttable}
\begin{tabular}{lccccc} \hline
 & (1) & (2) & (3) & (4) & (5) \\
 & Event & Event & Event & Event & Event \\ \hline
 &  &  &  &  &  \\
Diffusion Centrality & 0.607 &  &  &  &  \\
 & (0.085) &  &  &  &  \\
Degree Centrality &  & 0.460 &  &  &  \\
 &  & (0.078) &  &  &  \\
Eigenvector Centrality &  &  & 0.605 &  &  \\
 &  &  & (0.094) &  &  \\
Leader &  &  &  & 0.915 &  \\
 &  &  &  & (0.279) &  \\
Geographic Centrality &  &  &  &  & -0.082 \\
 &  &  &  &  & (0.136) \\
 &  &  &  &  &  \\
 Observations & 6,466 & 6,466 & 6,466 & 6,466 & 6,466 \\ \hline
\end{tabular}

\begin{tabular}{lccccc} \hline
 & (1) & (2) & (3) & (4) & (5) \\
 & Loan & Loan & Loan & Loan & Loan \\ \hline
 &  &  &  &  &  \\
Diffusion Centrality & 0.625 &  &  &  &  \\
 & (0.075) &  &  &  &  \\
Degree Centrality &  & 0.490 &  &  &  \\
 &  & (0.067) &  &  &  \\
Eigenvector Centrality &  &  & 0.614 &  &  \\
 &  &  & (0.084) &  &  \\
Leader &  &  &  & 1.013 &  \\
 &  &  &  & (0.263) &  \\
Geographic Centrality &  &  &  &  & -0.113 \\
 &  &  &  &  & (0.082) \\
 &  &  &  &  &  \\
 Observations & 6,466 & 6,466 & 6,466 & 6,466 & 6,466 \\ \hline
\end{tabular}

\begin{tablenotes}
Notes: This table uses data from the Phase 1 dataset. It reports estimates of Poisson regressions where the outcome variable is the expected number of nominations. Panel A presents results for the event question, and Panel B presents results for the loan question. 
Degree centrality, eigenvector centrality, and diffusion centrality, $DC\left(\g; 1/\E[\lambda_1],\E[Diam(\g(n,p))]\right)$, are normalized by their standard deviations. Standard errors (clustered at the village level) are reported in parentheses.
\end{tablenotes}
\end{threeparttable}
\end{table}

\clearpage

\begin{table}[!h]
\centering
\begin{threeparttable}
\caption{Factors predicting nominations}\label{tab:predict_biv}
\begin{tabular}{lcccccc} \hline
 & (1) & (2) & (3) & (4) & (5) & (6) \\
 & Event & Event & Event & Event & Event & Event \\ \hline
 &  &  &  &  &  &  \\
Diffusion Centrality & 0.642 & 0.354 & 0.567 & 0.606 & 0.374 & 0.607 \\
 & (0.127) & (0.176) & (0.091) & (0.085) & (0.206) & (0.085) \\
Degree Centrality & -0.039 &  &  &  & -0.020 &  \\
 & (0.101) &  &  &  & (0.101) &  \\
Eigenvector Centrality &  & 0.283 &  &  & 0.281 &  \\
 &  & (0.186) &  &  & (0.186) &  \\
Leader &  &  & 0.535 &  &  &  \\
 &  &  & (0.301) &  &  &  \\
Geographic Centrality &  &  &  & -0.082 &  &  \\
 &  &  &  & (0.142) &  &  \\
 &  &  &  &  &  &  \\
Observations & 6,466 & 6,466 & 6,466 & 6,466 & 6,466 & 6,466 \\
 Post-LASSO &  &  &  &  &  & \checkmark \\ \hline
\end{tabular}

\begin{tabular}{lcccccc} \hline
 & (1) & (2) & (3) & (4) & (5) & (6) \\
 & Loan & Loan & Loan & Loan & Loan & Loan \\ \hline
 &  &  &  &  &  &  \\
Diffusion Centrality & 0.560 & 0.431 & 0.578 & 0.624 & 0.339 & 0.560 \\
 & (0.122) & (0.130) & (0.081) & (0.075) & (0.170) & (0.122) \\
Degree Centrality & 0.070 &  &  &  & 0.088 & 0.070 \\
 & (0.086) &  &  &  & (0.084) & (0.086) \\
Eigenvector Centrality &  & 0.219 &  &  & 0.231 &  \\
 &  & (0.138) &  &  & (0.138) &  \\
Leader &  &  & 0.623 &  &  &  \\
 &  &  & (0.288) &  &  &  \\
Geographic Centrality &  &  &  & -0.115 &  &  \\
 &  &  &  & (0.089) &  &  \\
 &  &  &  &  &  &  \\
Observations & 6,466 & 6,466 & 6,466 & 6,466 & 6,466 & 6,466 \\
 Post-LASSO &  &  &  &  &  & \checkmark \\ \hline
\end{tabular}

\begin{tablenotes}
Notes: This table uses data from the Phase 1 dataset. It reports estimates of Poisson regressions where the outcome variable is the expected number of nominations under the event question. Panel A presents results for the event question, and Panel B presents results for the loan question. 
Degree centrality, eigenvector centrality, and diffusion centrality, $DC\left(\g; 1/\E[\lambda_1],\E[Diam(\g(n,p))]\right)$, are normalized by their standard deviations. Column (6) uses a post-LASSO procedure where in the first stage LASSO is implemented to select regressors and in the second stage the regression in question is run on those regressors. Omitted terms indicate they were not selected in the first stage. Standard errors (clustered at the village level) are reported in parentheses. 
\end{tablenotes} 
\end{threeparttable}
\end{table}

\clearpage

\begin{table}[!h]
\centering
\caption{Calls received by treatment} \label{tab:mainExpt_2_1}
\scalebox{0.8}{\begin{threeparttable}
\begin{tabular}{lccccc} \hline
 & (1) & (2) & (3) & (4) & (5) \\
 & RF & OLS & IV 1: First Stage & IV 2: First Stage & IV: Second Stage \\
 & Calls Received & Calls Received & At least 1 Gossip & At least 1 Elder & Calls Received \\ \hline
 &  &  &  &  &  \\
Gossip Treatment & 3.651 &  & 0.644 & 0.328 &  \\
 & (2.786) &  & (0.0660) & (0.0824) &  \\
Elder Treatment & -1.219 &  & 0.230 & 0.842 &  \\
 & (2.053) &  & (0.0807) & (0.0509) &  \\
At least 1 Gossip &  & 3.786 &  &  & 7.436 \\
 &  & (1.858) &  &  & (4.266) \\
At least 1 Elder &  & 0.792 &  &  & -3.475 \\
 &  & (2.056) &  &  & (2.259) \\
 &  &  &  &  &  \\
Observations & 212 & 212 & 212 & 212 & 212 \\
 Control Group Mean & 8.077 & 5.846 & 0.391 & 0.184 & 5.805 \\ \hline
\end{tabular}

\begin{tabular}{lccccc} \hline
 & (1) & (2) & (3) & (4) & (5) \\
 & RF & OLS & IV 1: First Stage & IV 2: First Stage & IV: Second Stage \\
 & $\frac{\mbox{Calls Received}}{\mbox{Seeds}}$ & $\frac{\mbox{Calls Received}}{\mbox{Seeds}}$ & At least 1 Gossip & At least 1 Elder & $\frac{\mbox{Calls Received}}{\mbox{Seeds}}$ \\ \hline
 &  &  &  &  &  \\
Gossip Treatment & 1.053 &  & 0.644 & 0.328 &  \\
 & (0.698) &  & (0.0660) & (0.0824) &  \\
Elder Treatment & -0.116 &  & 0.230 & 0.842 &  \\
 & (0.518) &  & (0.0807) & (0.0509) &  \\
At least 1 Gossip &  & 0.952 &  &  & 1.979 \\
 &  & (0.501) &  &  & (1.071) \\
At least 1 Elder &  & 0.309 &  &  & -0.677 \\
 &  & (0.511) &  &  & (0.588) \\
 &  &  &  &  &  \\
Observations & 212 & 212 & 212 & 212 & 212 \\
 Control Group Mean & 1.967 & 1.451 & 0.391 & 0.184 & 1.317 \\ \hline
\end{tabular}

\begin{tablenotes}Notes: This table uses data from the Phase 2 experimental dataset. Panel A uses the number of calls received as the outcome variable. Panel B normalizes the number of calls received by the number of seeds, 3 or 5, which is randomly assigned. For both panels, Column (1) shows the reduced form results of regressing number of calls received on dummies for gossip treatment and elder treatment. Column (2) regresses number of calls received on the  dummies for if at least 1 gossip was hit and for if at least 1 elder was hit in the village. Columns (3) and (4) show the first stages of the instrumental variable regressions, where the dummies for ``at least 1 gossip'' and ``at least 1 elder'' are regressed on the exogenous variables: gossip treatment dummy and elder treatment dummy. Column (5) shows the second stage of the IV; it regresses the number of calls received on the dummies for if at least 1 gossip was hit and if at least 1 elder was hit, both instrumented by treatment status of the village (gossip treatment or not, elder treatment or not). All columns control for number of gossips, number of elders, and number of seeds. For columns (1), (3), and (4)  the control group mean is calculated as the mean expectation of the outcome variable when the treatment is ``random''. For columns (2) and (5) the control group mean is calculated as the mean expectation of the outcome variable when no gossips or elders are reached. The control group mean for the second stage IV is calculated using IV estimates. Robust standard errors are reported in parentheses.
\end{tablenotes}
\end{threeparttable}} 
\end{table}

\clearpage
 
\setcounter{table}{5}
\renewcommand{\thetable}{\arabic{table}}

\begin{table}[!h]
\centering
\caption{Calls received by seed type}\label{tab:seedType_OLS}
\scalebox{0.65}{\begin{threeparttable}
\begin{tabular}{lcccccc} \hline
 & (1) & (2) & (3) & (4) & (5) & (6) \\
 & Calls Received & Calls Received & Calls Received & $\frac{\mbox{Calls Received}}{\mbox{Seeds}}$ & $\frac{\mbox{Calls Received}}{\mbox{Seeds}}$ & $\frac{\mbox{Calls Received}}{\mbox{Seeds}}$ \\ \hline
 &  &  &  &  &  &  \\
At least 1 Gossip & 6.645 & 5.574 &  & 1.637 & 1.370 &  \\
 & (3.867) & (4.119) &  & (0.949) & (0.992) &  \\
At least 1 Elder & 0.346 & 0.0566 &  & 0.245 & 0.173 &  \\
 & (3.602) & (3.576) &  & (0.926) & (0.912) &  \\
At least 1 High \textit{DC} Seed &  & 3.663 & 5.183 &  & 0.916 & 1.312 \\
 &  & (2.494) & (2.383) &  & (0.623) & (0.649) \\
 &  &  &  &  &  &  \\
Observations & 68 & 68 & 68 & 68 & 68 & 68 \\
 Control Group Mean & 5.586 & 5.586 & 5.719 & 1.353 & 1.353 & 1.402 \\ \hline
\end{tabular}

\begin{tablenotes}Notes: This table uses data from the Phase 2 experimental and network dataset. It presents OLS regressions of number of calls received (and number of calls received normalized by the number of seeds, 3 or 5, which is randomly assigned) on characteristics of the set of seeds. High $DC$ refers to a seed being above the mean by one standard deviation of the centrality distribution. All columns control for total number of gossips, number of elders, and number of seeds. For columns (1), (2), (4), and (5), the control group mean is calculated as the mean expectation of the outcome variable when no gossips or elders are reached. For columns (3) and (6), the control group mean is calculated as the mean expectation of the outcome variable when no high $DC$ seeds are reached. Robust standard errors are reported in parentheses. 
\end{tablenotes}
\end{threeparttable}
}
\end{table}

\newpage

\section{Threshold Parameters $(q,T)$ for Diffusion Centrality}\label{sec:Echoes}

We present two new theoretical results about diffusion centrality:
Theorem \ref{thm:bench}
and Corollary \ref{cor:bench}. We explicitly demonstrate that there
are natural intermediate parameters associated with diffusion centrality
at which it is distinct from the two boundary cases in which it simplifies
to other well-known centrality measures.

We can think of overall the number of times everyone is informed about
information coming from a seed as being composed of direct paths (the
seed, $i$, tells $j$), indirect natural paths ($i$ tells $j$ who
tells $k$ who tells $l$ and each are distinct), and echoes or other
cycles ($i$
tells $j$ who tells $k$ who tells $j$ who tells $k$). If there
is only one round of communication, then information never travels
beyond the seed's neighborhood. In that case diffusion centrality
just counts direct paths and is coincides with degree centrality.
On the other hand, if there are infinite rounds of communication (and
the probability of communicating across a link is high enough), diffusion
centrality converges to eigenvector centrality, and by capturing arbitrary
walks is partly driven by echoes and cycles as well as potentially long indirect paths.
Our proposed intermediate benchmark captures direct paths and indirect
natural paths and involves fewer cycles (which then become endemic as $T$ goes to
infinity). Theorem \ref{thm:bench}
and Corollary \ref{cor:bench} are theoretical results that provide
network-based guidance on what intermediate parameters achieve this
goal of mostly stripping out echoes.

Here we report a theorem and corollary that formalize some of these
intuitive statements. To do this, we consider a sequence of Erdos\textendash Renyi
networks, as those provide for clear limiting properties.

These properties extend to more general classes of random graph models
by standard arguments (e.g., see \citet{jackson2008b}), but an exploration
of such models takes us beyond our scope here.

Let ${\bf g}(n,p)$ denote an Erdos\textendash Renyi random network
drawn on $n$ nodes, with each link having independent probability
$p$. In the following, as is standard, $p$ (and $T$) are functions
of $n$, but we omit that notation to keep the expressions uncluttered.
We also allow for self-links for ease of calculations. We consider
a sequence of random graphs of size $n$ and as is standard in the
literature, consider what happens as $n\rightarrow\infty$.

\begin{thm} \label{thm:bench} If $T$ is not too large ($T=o({pn})$),\footnote{To remind the reader, $f(n)=o(h(n))$ for functions $f,h$ if $f(n)/h(n)\rightarrow0$,
and $f(n)=\Omega(h(n))$ if there exists $k>0$ for which $f(n)\geq kh(n)$
for all large enough $n$.} then the expected diffusion centrality of any node converges to $npq\frac{1-(npq)^{T}}{1-npq}$.
That is, for any $i$,
\[
\frac{\E\left[{DC}\left({\bf g}(n,p);q,T\right)_{i}\right]}{npq\frac{1-(npq)^{T}}{1-npq}}\rightarrow1.
\]
\end{thm}

Theorem \ref{thm:bench} provides a precise expression for how diffusion
centrality behaves in large graphs. Provided that $T$ grows at a
rate that is not overly fast\footnote{Note that $T$ can still grow at a rate that can tend to infinity
and in particular can grow faster than the growth rate of the diameter
of the network \textendash{} $T$ can grow up to $pn$, which will
generally be larger than $\log(n)$, while diameter is proportional
to $\log(n)/\log(pn)$.}, then we {\sl expect} the diffusion centrality of a typical node to converge
to $npq\frac{1-(npq)^{T}}{1-npq}$.  Of course, individual nodes vary in the centralities based on the realized network, but this result provides us with
the extent of diffusion that is expected from nodes, on average.

Theorem \ref{thm:bench} thus provides us with a tool to see when
a diffusion that begins at a typical node is expected to reach most other
nodes or not, on average, and leads to the following corollary.

\begin{cor} \label{cor:bench} Consider a sequence of Erdos-Renyi
random networks ${\bf g}(n,p)$ for which $\frac{1-\varepsilon}{\sqrt{n}}\geq p\geq(1+\varepsilon)\frac{\log(n)}{n}$
for some $\varepsilon>0$\footnote{This ensures that the network is connected almost surely as $n$ grows,
but not so dense that the diameter shrinks to be trivial. See \citet{bollobas2001}.} and any corresponding $T=o({pn})$. Then for any node $i$:
\begin{enumerate}
\item $1/\E[\lambda_{1}]$ is a threshold for $q$ as to whether diffusion
reaches a vanishing or expanding number of nodes :
\begin{enumerate}
\item If $q=o(1/\E[\lambda_{1}])$, then $\E\left[{DC}\left({\bf g}(n,p);q,T\right)_{i}\right]\rightarrow0$.
\item If $1/\E[\lambda_{1}]=o(q)$, then $\E\left[{DC}\left({\bf g}(n,p);q,T\right)_{i}\right]\rightarrow\infty$.\footnote{Note that $\E[\lambda_{1}]=np$.}
\end{enumerate}
\item $\E[Diam({\bf g}(n,p))]$ is a threshold relative for $T$ as to whether
diffusion reaches a vanishing or full fraction of nodes:\footnote{Again, note that $T=o(pn)$ is satisfied whenever $T=o(\log(n))$,
and thus is easily satisfied given that diameter is proportional to
$\log(n)/\log(pn)$ .}
\begin{enumerate}
\item If $T<(1-\varepsilon)\E[Diam({\bf g}(n,p))]$ for some $\varepsilon>0$,
then $\frac{\E\left[{DC}\left({\bf g}(n,p);q,T\right)_{i}\right]}{n}\rightarrow0$.
\item If $T\geq\E[Diam({\bf g}(n,p))]$ and $q>1/(\E[\lambda_{1}])^{1-\varepsilon}$
for some $\varepsilon>0$, then $\frac{\E\left[{DC}\left({\bf g}(n,p);q,T\right)_{i}\right]}{n}=\Omega(1)$.
\end{enumerate}
\end{enumerate}
\end{cor}

Putting these results together, we know that $q=1/\E[\lambda_{1}]$
and $T=\E[Diam(\g)]$ are the critical values where the process transitions
from a regime where diffusion is expected (in a large network) to
reach almost nobody to one where it will saturate the network. At
the critical value itself, diffusion reaches a non-trivial fraction
of the network but not everybody in it.

This makes $DC\left(\g;1/\E[\lambda_{1}],\E[Diam(\g(n,p))]\right)$
an interesting measure of centrality, distinct from other standard
measures of centrality at these values of the parameters. This fixes
$q$ and $T$ as a function of the graph so that the centrality measure
no longer has any free parameters \textendash{} enabling one to compare
it to other centrality measures without worrying that it performs
better simply because it has parameters that can be adjusted by the
researcher. For reasons explained earlier, we use it throughout
the empirical sections. 


\newpage
\begin{center}
{\bf \large{For Online Publication}}
\end{center}

\section{Proofs}\label{sec:proofs}

\subsection{Relation of Diffusion Centrality to Other Measures}

\

We prove all of the statements for the case of weighted ($g_{ij} \in [0,1]$) and directed networks ($g_{ij}$ can differ from $g_{ji}$).  Thus, we can take ${\bf g}\in [0,1]^{n\times n}$ and to allow
for full heterogeneity in communication.  For instance, $g_{ij}$ and $g_{ik}$ could both be positive, and yet differ from each other.
In what follows, we still include $q$ as an explicit multiplier, noting that this is redundant given the generality of the $\g$ matrix. The reader who finds this distracting can set $q=1$.

Let $v^{\left(L,k\right)}$ indicate $k$-th left-hand side eigenvector of ${\bf g}$ and similarly let $v^{\left(R,k\right)}$
indicate ${\bf g}$'s $k$-th right-hand side eigenvector.  In the case of undirected networks, $v^{\left(L,k\right)}=v^{\left(R,k\right)}$.

Let $d({\bf g})$ denote (out) degree centrality (so $d_{i}({\bf g})=\sum_{j}g_{ij}$).
Eigenvector centrality corresponds to $v^{(R,1)}({\bf g})$: the first
eigenvector of ${\bf g}$. Also, let $KB({\bf g},q)$ denote Katz--Bonacich
centrality -- defined for $q<1/\lambda_{1}$ by:%
\footnote{See (2.7) in \citet{jackson2008} for additional discussion and background.
This was a measure first discussed by Katz, and corresponds to Bonacich's
definition when both of Bonacich's parameters are set to $q$.%
}
\[
KB({\bf g},q):=\left(\sum_{t=1}^{\infty}\left(q{\bf g}\right)^{t}\right)\cdot{\bf 1}.
\]

It is direct to see that (i) diffusion centrality is proportional
to degree centrality at the extreme at which $T=1$, and (ii) if $q<1/\lambda_{1}$,
then diffusion centrality coincides with Katz--Bonacich centrality
if we set $T=\infty$. We now show that when $q>1/\lambda_{1}$ diffusion
centrality approaches eigenvector centrality as $T$ approaches $\infty$,
which then completes the picture of the relationship between diffusion
centrality and extreme centrality measures.

The difference between the extremes of Katz--Bonacich centrality and
eigenvector centrality depends on whether $q$ is sufficiently small
so that limited diffusion takes place even in the limit of large $T$,
or whether $q$ is sufficiently large so that the knowledge saturates
the network and then it is only relative amounts of saturation that
are being measured.%
\footnote{Saturation occurs when the entries of $\left(\sum_{t=1}^{\infty}\left(q{\bf g}\right)^{t}\right)\cdot{\bf 1}$
diverge (note that in a {[}strongly{]} connected network, if one entry
diverges, then all entries diverge). Nonetheless, the limit vector
is still proportional to a well defined limit vector: the first eigenvector.%
}

\

\begin{thm} \label{DC}

\

\begin{enumerate}
\item Diffusion centrality is proportional to (out) degree when $T=1$:
\[
{DC}\left({\bf g};q,1\right)=qd\left({\bf g}\right).
\]

\item If $q\geq 1/\lambda_{1}$ and $\g$ is aperiodic, then as $T\rightarrow\infty$ diffusion
centrality approximates eigenvector centrality:
\[
\lim_{T\rightarrow\infty}\frac{{DC}\left({\bf g};q,T\right)}{\sum_{t=1}^{T}\left(q\lambda_{1}\right)^{t}}=v^{(R,1)}.
\]

\item For $T=\infty$ and $q<1/\lambda_{1}$, diffusion centrality is Katz--Bonacich
centrality:
\[
{DC}\left({\bf g};q,\infty\right)=KB\left({\bf g},q\right);\ \ q<1/\lambda_{1}.
\]

\end{enumerate}
\end{thm}

This is a result we mention in \citet*{banerjeecdj2013}.
An independent formalization appears in \citet{benzik2014}.%

We also remark on the comparison to another measure: the influence
vector that appears in the DeGroot learning model (see, e.g., \citet{golubj2010}).
That metric captures how influential a node is in a process of social
learning. It corresponds to the (left-hand) unit eigenvector of a
stochasticized matrix of interactions rather than a raw adjacency
matrix. While it might be tempting to use that metric here as well,
we note that it is the right conceptual object to use in a process
of \emph{repeated averaging} through which individuals update opinions
based on averages of their neighbors' opinions. It is thus conceptually
different from the diffusion process that we study.
Nonetheless, one can also define
a variant of diffusion centrality that works for finite iterations
of DeGroot updating.

\

\begin{proof}[{\bf Proof of Theorem \ref{DC}}]  We show the second statement as the others follow directly.

First,  consider any irreducible and aperiodic nonnegative (and hence primitive) $\g$.  If the statement holds for
any arbitrarily close positive and diagonalizable $\g'$ (which are dense in a nonnegative neighborhood of $\g$), then since $\frac{{DC}\left({\bf g};q,T\right)}{\sum_{t=1}^T \left( q \lambda_1\right)^{t}}$ is a continuous function (in a neighborhood of a primitive $\g$, which has a simple first eigenvalue)
as is the first eigenvector, then the statement also holds at $\g$.\footnote{As is shown below,  $\frac{{DC}\left({\bf g};q,T\right)}{\sum_{t=1}^T \left( q \lambda_1\right)^{t}}$ has a well-defined limit, and so this holds also for the limit.  }
Thus, it is enough to prove the result for a positive and diagonalizable $\g$.

We show the following  for a positive and diagonalizable $\g$:
\begin{itemize}

\item  If $q> \lambda_1^{-1}$, then
$$\lim_{T\rightarrow\infty}\frac{{DC}\left({\bf g};q,T\right)}{\sum_{t=1}^T \left( q \lambda_1\right)^{t}}= \lim_{T\rightarrow\infty} \frac{{DC}\left(g;q,T\right)}{\frac{ q \lambda_1-(q \lambda_1)^{T+1}}{1-(q \lambda_1)}}= {v}^{\left(R,1\right)}.$$

\item If $q = \lambda_1^{-1}$, then

$$\lim_{T\rightarrow\infty}\frac{1}{T}{DC}\left({\bf g};\lambda_{1}^{-1},T\right) = {v}^{\left(R,1\right)}.$$
\end{itemize}

\

Let $\widetilde{{\bf g}}={\bf g}/\lambda_{1}$, and normalize the eigenvectors to lie in $\ell_{1}$,
so that the entries in each column of ${\bf V}^{-1}$ and each row of ${\bf V}$ sum to 1.

Let us show the statement for the case where $q=1/\lambda_1$. It is sufficient to show
\[
\lim_{T\rightarrow\infty}\left\Vert \frac{DC\left(\g;\lambda_{1}^{-1},T\right)}{T}-{v}^{(R,1)}\right\Vert =0.
\]
First, note that given the diagonalizable matrix, straightforward calculations show that
\[
DC_{i}\left(\g;\lambda_{1}^{-1},T\right) =\sum_j \sum_{t=1}^T \sum_k  v_i^{(R,k)} v_j^{(L,k)} \widetilde{\lambda}_{k}^{t}.
\]
Thus,
\begin{eqnarray*}
\left|\frac{DC_{i}\left(\g;\lambda_{1}^{-1},T\right)}{T}-v_{i}^{(R,1)}\right| & = &
\left|\frac{ \sum_j \sum_{t=1}^T \sum_{k=1}^n  v_i^{(R,k)} v_j^{(L,k)} \widetilde{\lambda}_{k}^{t}}{T} -v_{i}^{(R,1)} \right| =  \\
& = & \left|\frac{1}{T} \sum_j \sum_{t=1}^T \sum_{k= 2}^n  v_i^{(R,k)} v_j^{(L,k)} \widetilde{\lambda}_{k}^{t} \right|\leq\frac{1}{T}\sum_{t=1}^{T}\sum_{k=2}^{n}1\cdot\underbrace{\left|\sum_{j=1}^{n}v_{j}^{(L,k)}\right|}_{\leq1}
\cdot\left|\widetilde{\lambda}_{k}^{t}\right|\\
 & \leq & \frac{n}{T}\sum_{t=1}^{T}\left|\widetilde{\lambda}_{2}^{t}\right|=\frac{n}{T}\frac{\left|\widetilde{\lambda}_{2}\right|}{1-\left|\widetilde{\lambda}_{2}\right|}\left(1-\left|\widetilde{\lambda}_{2}\right|^{T}\right)\rightarrow0.
\end{eqnarray*}
Since the length of the vector (which is $n$) is unchanging in $T$, pointwise
convergence implies convergence in norm, proving the result.

The final piece repeats the argument for $q>1/\lambda_1$.   It follows that the eigenvalues of $q{\bf g}$ are $\widetilde{\Lambda} = {\rm diag}\left\{\widetilde{\lambda}_1,...,\widetilde{\lambda}_n\right\}$ with $q\lambda_k = \widetilde{\lambda}_k$.  We show
\[
\lim_{T\rightarrow\infty}\left\Vert \frac{DC\left(\g;q,T\right)}{\sum_{t=1}^{T}\left(q\lambda_{1}\right)^{t}}-{v}^{(R,1)}\right\Vert =0.
\]
By similar derivations as above,
\begin{eqnarray*}
\left|\frac{DC_{i}\left(\g;\lambda_{1}^{-1},T\right)}{\sum_{t=1}^{T}\widetilde{\lambda}_{1}^{t}}-v_{i}^{(R,1)}\right| & = &
\left|\frac{ \sum_j \sum_{t=1}^T \sum_{k=1}^n  v_i^{(R,k)} v_j^{(L,k)} \widetilde{\lambda}_{k}^{t}}{\sum_{t=1}^{T}\widetilde{\lambda}_{1}^{t}} -v_{i}^{(R,1)} \right|   \\
 & = &
\left|\frac{ \sum_j \sum_{t=1}^T \sum_{k=2}^n  v_i^{(R,k)} v_j^{(L,k)} \widetilde{\lambda}_{k}^{t}}{\sum_{t=1}^{T}\widetilde{\lambda}_{1}^{t}}  +
\frac{\sum_j \sum_{t=1}^T  v_i^{(R,1)} v_j^{(L,1)}\widetilde{\lambda}_{1}^{t}}{\sum_{t=1}^{T}\widetilde{\lambda}_{1}^{t}}
-v_{i}^{(R,1)} \right|   \\
 & = &
\left|\frac{ \sum_j \sum_{t=1}^T \sum_{k=2}^n  v_i^{(R,k)} v_j^{(L,k)} \widetilde{\lambda}_{k}^{t}}{\sum_{t=1}^{T}\widetilde{\lambda}_{1}^{t}}  +
\frac{ \sum_{t=1}^T  v_i^{(R,1)} \widetilde{\lambda}_{1}^{t}}{\sum_{t=1}^{T}\widetilde{\lambda}_{1}^{t}}
-v_{i}^{(R,1)} \right|   \\
& = & \left|\frac{1}{\sum_{t=1}^{T}\widetilde{\lambda}_{1}^{t}} \sum_j \sum_{t=1}^T \sum_{k= 2}^n  v_i^{(R,k)} v_j^{(L,k)} \widetilde{\lambda}_{k}^{t} \right|\\
& \leq &\frac{1}{\sum_{t=1}^{T}\widetilde{\lambda}_{1}^{t}}\sum_{t=1}^{T}\sum_{k=2}^{n}1\cdot
\left|\sum_{j=1}^{n}v_{j}^{(L,k)}\right|
\cdot\left|\widetilde{\lambda}_{k}^{t}\right|\\
 & \leq & \frac{n}{\sum_{t=1}^{T}\widetilde{\lambda}_{1}^{t}}\sum_{t=1}^{T}\left|\widetilde{\lambda}_{2}^{t}\right|.
\end{eqnarray*}
Note that this last expression converges to 0 since $\widetilde{\lambda}_{1}>1$, and
$\widetilde{\lambda}_{1}>\widetilde{\lambda}_{2}$.\footnote{\label{fn:domination}Note that it is important
that $q\geq 1/\lambda_1$ for this claim, since if $q < 1/ \lambda_1$, then $q \lambda_1 < 1$. In that case, observe that
$$ \frac{\sum_{t=1}^{T}\left|\tilde{\lambda}_{2}\right|^{t}}{\sum_{t=1}^{T}\tilde{\lambda}_{1}^{t}} = \frac{\widetilde{\lambda}_2}{\widetilde{\lambda}_1} \cdot \frac{1-\tilde{\lambda_1}}{1-\tilde{\lambda_2}}$$ by the properties of a geometric sum, which is of constant order.  Thus, higher order terms ($\widetilde{\lambda}_2$, etc.) persistently matter and are not dominated relative to $\sum_t^T \widetilde{\lambda}_1^t$.}
which completes the argument.
\end{proof}

\subsection{Other Proofs}

\

\begin{proof}[ {\bf Proof of Theorem \ref{thm:bench}} ]
\begin{eqnarray*}
\E\left[ {DC}\left({\bf g}(n,p);q,T\right)\right]_i & = & \left[\sum_1^T \E\left[ q^t {\bf g}(n,p)^t\right]\cdot {\bf 1}\right]_i\cr
  & = & \sum_1^T q^t n \E\left[  {\bf g}(n,p)^t\right]_{ij},
\end{eqnarray*}
where the last equality comes from the fact that $\E\left[  {\bf g}(n,p)^t\right]_{ij}=\E\left[  {\bf g}(n,p)^t\right]_{ik}$ for all $i,j,k$ in an Erdos--Renyi random graph.

Next, note that
\[
\E\left[  {\bf g}(n,p)^t\right]_{ij}=
\E\left[   \sum_{k_1, k_2, \ldots, k_{t-1} \in \{1,\ldots,n \}^{t-1} }  g_{ik_1} g_{k_1k_2} \cdots g_{k_{t-1}j} \right]
\]
If all the indexed $g_{..}$'s were distinct, the right hand side of this equation would simply be
$n^{t-1}p^t$. However, in the summand sometimes terms repeat. For example, if there were exactly $x$ repetitions, the probability of getting the walk would be $p^{t-x}$ instead of $p^t$. Thus, it follows directly that
\[
\E\left[  {\bf g}(n,p)^t\right]_{ij}\geq n^{t-1} p^t
\]
and so
\begin{eqnarray*}
\E\left[ {DC}\left({\bf g}(n,p);q,T\right)\right]_i & = & \sum_1^T q^t n \E\left[  {\bf g}(n,p)^t\right]_{ij}\cr
& \geq & \sum_1^T q^t n^t p^t  = npq \frac{1-(npq)^T}{1-npq}
\end{eqnarray*}

Note also, that
\[
\E\left[   \sum_{k_1, k_2, \ldots, k_t \in \{1,\ldots,n \}^t }  g_{ik_1} g_{k_1k_2} \cdots g_{k_{t-1} j} \right]
\leq n^{t-1} p^t + t n^{t-2} p^{t-1} + t^2 n^{t-3} p^{t-2} + \ldots + t^t .
\]
This last inequality is a very loose upper bound generated by setting a loose upper bound on how many $g_{..}$'s could conceivably repeat, and then putting in the expression that would ensue if
they did repeat.  Despite how loose the bound is, it suffices for our purposes.

Given that $t\leq T< pn$, it follows that
\begin{align*}
\E\left[   \sum_{k_1, k_2, \ldots, k_t \in \{1,\ldots,n \}^t }  g_{ik_1} g_{k_1k_2} \cdots g_{k_{t-1} j} \right]
& \leq n^{t-1} p^t \left(1+ \frac{t}{pn} + \left(\frac{t}{pn}\right)^2 \ldots + \left(\frac{t}{pn}\right)^t\right) \\
& =  n^{t-1} p^t \left(\frac{1 - \left(\frac{t}{pn}\right)^t}{1- \left(\frac{t}{pn}\right)}\right).
\end{align*}
Thus,
\[
\E\left[  {\bf g}(n,p)^t\right]_{ij}\leq n^{t-1} p^t \frac{ 1}{1- \frac{T}{pn}}.
\]

Since $ T<< {pn}$
it follows that (here $o(1)$ is with respect to $n$):
\begin{eqnarray*}
\E\left[ {DC}\left({\bf g}(n,p);q,T\right)\right]_i & = & \sum_1^T q^t n \E\left[  {\bf g}(n,p)^t\right]_{ij}\cr
& \leq & \sum_1^T q^t  n^t p^t (1+o(1))  = npq \frac{1-(npq)^T}{1-npq} (1+ o(1)).
\end{eqnarray*}

The theorem follows directly.  \end{proof}

\

\begin{proof}[ {\bf Proof of Theorem \ref{prop:cov}} ]
Recall that
$
\bfH= \sum_{t=1}^{T}\left(q{\bf g}\right)^{t}
$ and
$
DC =
\left(\sum_{t=1}^{T}\left(q{\bf g}\right)^{t}\right)\cdot{\bf 1}
$
and so
$$
DC_i = \sum_{j}H_{ij}.
$$

Additionally,
$$\cov (DC,H_{\cdot,j}) =
\sum_i  \left(DC_i- \sum_k \frac{DC_k}{n}\right)\left(H_{ij} - \sum_k \frac{H_{kj}}{n}\right).
$$
Thus
$$\sum_j\cov (DC,H_{\cdot,j}) =
\sum_i  \left(DC_i- \sum_k \frac{DC_k}{n}\right)\left(\sum_j H_{ij} - \sum_k \frac{\sum_j H_{kj}}{n}\right),
$$
implying
$$\sum_j\cov (DC,H_{\cdot,j}) =
\sum_i  \left(DC_i- \sum_k \frac{DC_k}{n}\right)\left(DC_i- \sum_k \frac{DC_k}{n}\right)
=\var(DC),
$$
which completes the proof.
\end{proof}

\

\begin{proof}[ {\bf Proof of Corollary \ref{cor:bench}} ]
To see (1), first note that
$x \frac{1-x^T}{1-x} \rightarrow 0 $ if $x\rightarrow 0$,
and that
$x \frac{1-x^T}{1-x} \rightarrow x \frac{x^T}{x}   \rightarrow \infty $ if $x\rightarrow \infty$.
Replacing $x$ with $npq$ and then applying Theorem \ref{thm:bench} yields the result under (a) and (b), respectively.

To see (2), we consider the case in which $q> 1/(\E[\lambda_1])^{1-\varepsilon}$, which of course is equivalent to $npq > (np)^\varepsilon$.
This is the case under which (b) applies.  This also implies the result in (a), since if the conclusion of (a) holds for such a $q$ it will also hold for all lower $q$, given that $DC$ is monotone in $q$.

Again, since $npq>1$, it follows that if $T$ is growing, then
\[
\E\left[ {DC}\left({\bf g}(n,p);q,T\right)\right]_i
\rightarrow
npq \frac{1-(npq)^T}{1-npq}
\rightarrow (npq)^T.
\]
So, to have
\[
\E\left[ {DC}\left({\bf g}(n,p);q,T\right)\right]_i \geq kn
\]
for some $k>0$, it is sufficient that $(npq)^T\geq kn $,
or
\[
T\geq \frac{\log(n)+\log(k)}{\log{np} + \log(q)} \rightarrow \frac{\log(n)}{\log{np}} \sim \E[Diam \left({\bf g}(n,p)\right)],
\]
where the last comparison is a property of Erdos--Renyi random networks given that $\frac{1-\varepsilon}{\sqrt{n}}\geq p\geq (1+\varepsilon) \frac{\log(n)}{n}$, and so this establishes (b).
From the analogous calculation, if $T$ is below $\frac{\log(n)}{\log{np}}$, then $\E\left[ {DC}\left({\bf g}(n,p);q,T\right)\right]_i \leq kn$ for any $k$, and so (a) follows.
\end{proof}

\

\begin{proof}[{\bf Proof of Theorem \ref{prop:rank}}]
Again, we prove the result for a positive diagonalizable $\g$, noting that it then holds for any (nonnegative) $\g$.

Again, let $\g$ be written as
\[
\g={\bf V}\Lambda {\bf V}^{-1}.
\]
Also, let $\tilde{\lambda}_{k}=q\lambda_{k}$.
It then follows that we can write
\[
{\bf H}=\sum_{t=1}^{T}\left(q{\bf g}\right)^{t} =\sum_{t=1}^{T}\left(\sum_{k=1}^n v_{i}^{\left(R,k\right)}v_{j}^{\left(L,k\right)}\widetilde{\lambda}_{k}^{t}\right).
\]
By the ordering of the eigenvalues from largest to smallest in magnitude,
\begin{eqnarray*}
{\bf H}_{\cdot,j} & = & \sum_{t=1}^{T}\left[{ v}^{(R,1)}v_{j}^{(L,1)}\tilde{\lambda}_{1}^{t}+{ v}^{(R,2)}v_{j}^{(L,2)}\tilde{\lambda}_{2}^{t}+O\left(\left|\tilde{\lambda}_{2}\right|^{t}\right)\right]\\
 & = & \sum_{t=1}^{T}\left[{ v}^{(R,1)}v_{j}^{(L,1)}\tilde{\lambda}_{1}^{t}+O\left(\left|\tilde{\lambda}_{2}\right|^{t}\right)\right]\\
 & = & { v}^{(R,1)}v_{j}^{(L,1)}\sum_{t=1}^{T}\tilde{\lambda}_{1}^{t}+O\left(\sum_{t=1}^{T}\left|\tilde{\lambda}_{2}\right|^{t}\right).
\end{eqnarray*}
So, since the largest eigenvalue is unique, it follows that
\[
\frac{{\bf H}_{\cdot,j}}{\sum_{t=1}^{T}\tilde{\lambda}_{1}^{t}}={ v}^{(R,1)}v_{j}^{(L,1)}+O\left(\frac{\sum_{t=1}^{T}\left|\tilde{\lambda}_{2}\right|^{t}}{\sum_{t=1}^{T}\tilde{\lambda}_{1}^{t}}\right).
\]
Note that the last expression converges to 0 since $\widetilde{\lambda}_{1}>1$, and
$\widetilde{\lambda}_{1}>\widetilde{\lambda}_{2}$.
Thus,
\[
\frac{{\bf H}_{\cdot,j}}{\sum_{t=1}^{T}\tilde{\lambda}_{1}^{t}}\rightarrow { v}^{(R,1)}v_{j}^{(L,1)}
\]
for each $j$.
This completes the proof since each column of ${\bf H}$ is proportional to ${ v}^{(R,1)}$ in the limit, and thus has the correct ranking for large enough $T$.\footnote{The discussion in Footnote \ref{fn:domination} clarifies why $q > 1/\lambda_1$ is required for the argument.}
Note that the ranking is up to ties, as the ranking of tied entries may vary arbitrarily along the sequence.  That is, if
${ v}^{(R,1)}_{i}={ v}^{(R,1)}_{\ell}$, then $j$'s ranking over $i$ and $\ell$ could vary arbitrarily with $T$, but their rankings will be correct relative to any other entries with higher or lower eigenvector centralities.
\end{proof}

\newpage

\section{Extension of Phase 1 results}\label{sec:phase1_ext}
This section extends the descriptive analysis from the Phase 1 network data on 33 villages. 
We repeat all of our analyses with OLS specifications instead of Poisson specifications. Additionally, we include a Post-LASSO estimation which conducts a LASSO to select which variables best explain our outcome of interest (number of nominations) and then does a post-estimation to recover consistent parameter estimates.

\setcounter{table}{0}
\renewcommand{\thetable}{B.\arabic{table}}

\begin{table}[!h]
\centering
\caption{Factors predicting nominations}\label{tab:predict1_OLS}
\scalebox{0.9}{\begin{threeparttable}
\begin{tabular}{lccccc} \hline
 & (1) & (2) & (3) & (4) & (5) \\
 & Event & Event & Event & Event & Event \\ \hline
 &  &  &  &  &  \\
Diffusion Centrality & 0.285 &  &  &  &  \\
 & (0.060) &  &  &  &  \\
Degree Centrality &  & 0.250 &  &  &  \\
 &  & (0.061) &  &  &  \\
Eigenvector Centrality &  &  & 0.283 &  &  \\
 &  &  & (0.064) &  &  \\
Leader &  &  &  & 0.436 &  \\
 &  &  &  & (0.168) &  \\
Geographic Centrality &  &  &  &  & -0.025 \\
 &  &  &  &  & (0.038) \\
 &  &  &  &  &  \\
 Observations & 6,466 & 6,466 & 6,466 & 6,466 & 6,466 \\ \hline
\end{tabular}

\begin{tabular}{lccccc} \hline
 & (1) & (2) & (3) & (4) & (5) \\
 & Loan & Loan & Loan & Loan & Loan \\ \hline
 &  &  &  &  &  \\
Diffusion Centrality & 0.391 &  &  &  &  \\
 & (0.071) &  &  &  &  \\
Degree Centrality &  & 0.367 &  &  &  \\
 &  & (0.065) &  &  &  \\
Eigenvector Centrality &  &  & 0.378 &  &  \\
 &  &  & (0.074) &  &  \\
Leader &  &  &  & 0.653 &  \\
 &  &  &  & (0.224) &  \\
Geographic Centrality &  &  &  &  & -0.045 \\
 &  &  &  &  & (0.029) \\
 &  &  &  &  &  \\
 Observations & 6,466 & 6,466 & 6,466 & 6,466 & 6,466 \\ \hline
\end{tabular}

\begin{tablenotes}Notes: This table uses data from the Phase 1 dataset. It reports estimates of OLS regressions where the outcome variable is the expected number of nominations under the event question. Panel A presents results for the event question, and Panel B presents results for the loan question. 
Degree centrality, eigenvector centrality, and diffusion centrality, $DC\left(\g; 1/\E[\lambda_1],\E[Diam(\g(n,p))]\right)$, are normalized by their standard deviations. Standard errors (clustered at the village level) are reported in parentheses.
\end{tablenotes}
\end{threeparttable}}
\end{table}

\clearpage

\begin{table}[!h]
\centering
\scalebox{0.9}{\begin{threeparttable}
\caption{Factors predicting nominations}\label{tab:predict_biv_OLS}
\begin{tabular}{lcccccc} \hline
 & (1) & (2) & (3) & (4) & (5) & (6) \\
 & Event & Event & Event & Event & Event & Event \\ \hline
 &  &  &  &  &  &  \\
Diffusion Centrality & 0.303 & 0.161 & 0.269 & 0.285 & 0.173 & 0.285 \\
 & (0.091) & (0.087) & (0.061) & (0.060) & (0.107) & (0.060) \\
Degree Centrality & -0.020 &  &  &  & -0.013 &  \\
 & (0.066) &  &  &  & (0.068) &  \\
Eigenvector Centrality &  & 0.138 &  &  & 0.137 &  \\
 &  & (0.095) &  &  & (0.095) &  \\
Leader &  &  & 0.294 &  &  &  \\
 &  &  & (0.174) &  &  &  \\
Geographic Centrality &  &  &  & -0.026 &  &  \\
 &  &  &  & (0.039) &  &  \\
 &  &  &  &  &  &  \\
Observations & 6,466 & 6,466 & 6,466 & 6,466 & 6,466 & 6,466 \\
 Post-LASSO &  &  &  &  &  & \checkmark \\ \hline
\end{tabular}

\begin{tabular}{lcccccc} \hline
 & (1) & (2) & (3) & (4) & (5) & (6) \\
 & Loan & Loan & Loan & Loan & Loan & Loan \\ \hline
 &  &  &  &  &  &  \\
Diffusion Centrality & 0.310 & 0.266 & 0.366 & 0.391 & 0.175 & 0.310 \\
 & (0.112) & (0.089) & (0.071) & (0.071) & (0.124) & (0.112) \\
Degree Centrality & 0.091 &  &  &  & 0.098 & 0.091 \\
 & (0.079) &  &  &  & (0.079) & (0.079) \\
Eigenvector Centrality &  & 0.138 &  &  & 0.144 &  \\
 &  & (0.089) &  &  & (0.087) &  \\
Leader &  &  & 0.461 &  &  &  \\
 &  &  & (0.229) &  &  &  \\
Geographic Centrality &  &  &  & -0.045 &  &  \\
 &  &  &  & (0.030) &  &  \\
 &  &  &  &  &  &  \\
Observations & 6,466 & 6,466 & 6,466 & 6,466 & 6,466 & 6,466 \\
 Post-LASSO &  &  &  &  &  & \checkmark \\ \hline
\end{tabular}

\begin{tablenotes}
Notes: This table uses data from the Phase 1 dataset. It reports estimates of OLS regressions where the outcome variable is the expected number of nominations. Panel A presents results for the event question, and Panel B presents results for the loan question. 
Degree centrality, eigenvector centrality, and diffusion centrality, $DC\left(\g; 1/\E[\lambda_1],\E[Diam(\g(n,p))]\right)$, are normalized by their standard deviations. Column (6) uses a post-LASSO procedure where in the first stage LASSO is implemented to select regressors and in the second stage the regression in question is run on those regressors. Omitted terms indicate they were not selected in the first stage. Standard errors (clustered at the village level) are reported in parentheses.
\end{tablenotes}
\end{threeparttable}}
\end{table}

\clearpage

\section{Extension of experiment analysis}\label{sec:expt_ext}
This section extends the analysis of the experiment results to using four instruments. 

\setcounter{table}{0}
\renewcommand{\thetable}{C.\arabic{table}}

\clearpage
\begin{table}[!h]
\centering
\caption{Calls received by treatment} \label{tab:mainExpt_2_2}
\scalebox{0.8}{\begin{threeparttable}
\begin{tabular}{lccccc} \hline
 & (1) & (2) & (3) & (4) & (5) \\
 & RF & OLS & IV 1: First Stage & IV 2: First Stage & IV: Second Stage \\
 & Calls Received & Calls Received & At least 1 Gossip & At least 1 Elder & Calls Received \\ \hline
 &  &  &  &  &  \\
Gossip Treatment & 4.559 &  & 0.795 & 0.430 &  \\
 & (3.121) &  & (0.0753) & (0.108) &  \\
5 Gossip Seeds & -1.785 &  & -0.303 & -0.206 &  \\
 & (5.290) &  & (0.110) & (0.153) &  \\
Elder Treatment & 2.279 &  & 0.370 & 0.872 &  \\
 & (2.424) &  & (0.106) & (0.0685) &  \\
5 Elder Seeds & -6.798 &  & -0.272 & -0.0578 &  \\
 & (3.487) &  & (0.149) & (0.100) &  \\
At least 1 Gossip &  & 3.786 &  &  & 8.063 \\
 &  & (1.858) &  &  & (3.845) \\
At least 1 Elder &  & 0.792 &  &  & -3.684 \\
 &  & (2.056) &  &  & (2.266) \\
 &  &  &  &  &  \\
Observations & 212 & 212 & 212 & 212 & 212 \\
 Control Group Mean & 8.019 & 5.846 & 0.389 & 0.183 & 5.496 \\ \hline
\end{tabular}

\begin{tabular}{lccccc} \hline
 & (1) & (2) & (3) & (4) & (5) \\
 & RF & OLS & IV 1: First Stage & IV 2: First Stage & IV: Second Stage \\
 & $\frac{\mbox{Calls Received}}{\mbox{Seeds}}$ & $\frac{\mbox{Calls Received}}{\mbox{Seeds}}$ & At least 1 Gossip & At least 1 Elder & $\frac{\mbox{Calls Received}}{\mbox{Seeds}}$ \\ \hline
 &  &  &  &  &  \\
Gossip Treatment & 1.593 &  & 0.795 & 0.430 &  \\
 & (1.030) &  & (0.0753) & (0.108) &  \\
5 Gossip Seeds & -1.083 &  & -0.303 & -0.206 &  \\
 & (1.348) &  & (0.110) & (0.153) &  \\
Elder Treatment & 0.622 &  & 0.370 & 0.872 &  \\
 & (0.770) &  & (0.106) & (0.0685) &  \\
5 Elder Seeds & -1.430 &  & -0.272 & -0.0578 &  \\
 & (0.912) &  & (0.149) & (0.100) &  \\
At least 1 Gossip &  & 0.952 &  &  & 2.169 \\
 &  & (0.501) &  &  & (1.043) \\
At least 1 Elder &  & 0.309 &  &  & -0.676 \\
 &  & (0.511) &  &  & (0.578) \\
 &  &  &  &  &  \\
Observations & 212 & 212 & 212 & 212 & 212 \\
 Control Group Mean & 1.953 & 1.451 & 0.389 & 0.183 & 1.186 \\ \hline
\end{tabular}

\begin{tablenotes}Notes: This table uses data from the Phase 2 experimental dataset. Panel A uses the number of calls received as the outcome variable. Panel B normalizes the number of calls received by the number of seeds, 3 or 5, which is randomly assigned. For both panels, Column (1) shows the reduced form results of regressing number of calls received on dummies for gossip treatment and elder treatment. Column (2) regresses number of calls received on the dummies for if at least 1 gossip was hit and for if at least 1 elder was hit in the village. Columns (3) and (4) show the first stages of the instrumental variable regressions, where the dummies for ``at least 1 gossip'' and ``at least 1 elder'' are regressed on the exogenous variables: gossip treatment dummy, 5 gossip seeds dummy, elder treatment dummy, 5 elder seeds dummy. Column (5) shows the second stage of the IV; it regresses the number of calls received on the dummies for if at least 1 gossip was hit and if at least 1 elder was hit, both instrumented by treatment status of the village (gossip treatment or not, elder treatment or not) and seed number dummies for the village (5 gossip seeds or not, 5 elder seeds or not). All columns control for number of gossips, number of elders and number of seeds. For columns (1), (3), and (4)  the control group mean is calculated as the mean expectation of the outcome variable when the treatment is ``random''. For columns (2) and (5), the control group mean is calculated as the mean expectation of the outcome variable when no gossips or elders are reached. The control group mean for the second stage IV is calculated using IV estimates. Robust standard errors are reported in parentheses.
\end{tablenotes}
\end{threeparttable}} 
\end{table}

\clearpage

\section{Experiment Analysis with Broadcast Village}\label{sec:broadcast}
This section repeats our main experimental analyses but includes the broadcast village where the poster was made by one of the seeds.

\setcounter{table}{0}
\renewcommand{\thetable}{D.\arabic{table}}

\begin{table}[!h]
\centering
\caption{Calls received by treatment}\label{tab:mainExpt_flyer}
\scalebox{0.8}{\begin{threeparttable}
\begin{tabular}{lccccc} \hline
 & (1) & (2) & (3) & (4) & (5) \\
 & RF & OLS & IV 1: First Stage & IV 2: First Stage & IV: Second Stage \\
 & Calls Received & Calls Received & At least 1 Gossip & At least 1 Elder & Calls Received \\ \hline
 &  &  &  &  &  \\
Gossip Treatment & 2.266 &  & 0.636 & 0.331 &  \\
 & (3.116) &  & (0.0660) & (0.0821) &  \\
Elder Treatment & -2.809 &  & 0.220 & 0.846 &  \\
 & (2.577) &  & (0.0807) & (0.0502) &  \\
At least 1 Gossip &  & 5.005 &  &  & 6.122 \\
 &  & (2.210) &  &  & (4.532) \\
At least 1 Elder &  & -0.619 &  &  & -4.914 \\
 &  & (2.472) &  &  & (2.628) \\
 &  &  &  &  &  \\
Observations & 213 & 213 & 213 & 213 & 213 \\
 Control Group Mean & 9.534 & 6.277 & 0.400 & 0.180 & 7.971 \\ \hline
\end{tabular}

\begin{tabular}{lccccc} \hline
 & (1) & (2) & (3) & (4) & (5) \\
 & RF & OLS & IV 1: First Stage & IV 2: First Stage & IV: Second Stage \\
 & $\frac{\mbox{Calls Received}}{\mbox{Seeds}}$ & $\frac{\mbox{Calls Received}}{\mbox{Seeds}}$ & At least 1 Gossip & At least 1 Elder & $\frac{\mbox{Calls Received}}{\mbox{Seeds}}$ \\ \hline
 &  &  &  &  &  \\
Gossip Treatment & 0.591 &  & 0.636 & 0.331 &  \\
 & (0.841) &  & (0.0660) & (0.0821) &  \\
Elder Treatment & -0.646 &  & 0.220 & 0.846 &  \\
 & (0.738) &  & (0.0807) & (0.0502) &  \\
At least 1 Gossip &  & 1.359 &  &  & 1.535 \\
 &  & (0.644) &  &  & (1.179) \\
At least 1 Elder &  & -0.162 &  &  & -1.164 \\
 &  & (0.691) &  &  & (0.748) \\
Constant &  &  &  & 0.109 &  \\
 &  &  &  & (0.160) &  \\
 &  &  &  &  &  \\
Observations & 213 & 213 & 213 & 213 & 213 \\
 Control Group Mean & 2.452 & 1.595 & 0.400 & 0.180 & 2.048 \\ \hline
\end{tabular}

\begin{tablenotes}Notes: This table uses data from the Phase 2 experimental dataset. Panel A uses the number of calls received as the outcome variable. Panel B normalizes the number of calls received by the number of seeds, 3 or 5, which is randomly assigned. For both panels, Column (1) shows the reduced form results of regressing number of calls received on dummies for gossip treatment and elder treatment. Column (2) regresses number of calls received on the dummies for if at least 1 gossip was hit and for if at least 1 elder was hit in the village. Columns (3) and (4) show the first stages of the instrumental variable regressions, where the dummies for ``at least 1 gossip'' and ``at least 1 elder'' are regressed on the exogenous variables: gossip treatment dummy and elder treatment dummy. Column (5) shows the second stage of the IV; it regresses the number of calls received on the dummies for if at least 1 gossip was hit and if at least 1 elder was hit, both instrumented by treatment status of the village (gossip treatment or not, elder treatment or not). All columns control for number of gossips, number of elders, and number of seeds. For columns (1), (3), and (4)  the control group mean is calculated as the mean expectation of the outcome variable when the treatment is ``random''. For columns (2) and (5), the control group mean is calculated as the mean expectation of the outcome variable when no gossips or elders are reached. The control group mean for the second stage IV is calculated using IV estimates. Robust standard errors are reported in parentheses.
\end{tablenotes}
\end{threeparttable}}
\end{table}

\clearpage

\begin{table}[!h]
\centering
\caption{Calls received by seed type}\label{tab:seedType_OLS_flyer}
\scalebox{0.65}{\begin{threeparttable}
\begin{tabular}{lcccccc} \hline
 & (1) & (2) & (3) & (4) & (5) & (6) \\
 & Calls Received & Calls Received & Calls Received & $\frac{\mbox{Calls Received}}{\mbox{Seeds}}$ & $\frac{\mbox{Calls Received}}{\mbox{Seeds}}$ & $\frac{\mbox{Calls Received}}{\mbox{Seeds}}$ \\ \hline
 &  &  &  &  &  &  \\
At least 1 Gossip & 12.89 & 13.02 &  & 3.751 & 3.871 &  \\
 & (7.225) & (8.157) &  & (2.282) & (2.584) &  \\
At least 1 Elder & -3.371 & -3.321 &  & -1.012 & -0.962 &  \\
 & (5.155) & (4.946) &  & (1.547) & (1.456) &  \\
At least 1 High \textit{DC} Seed &  & -0.485 & 2.262 &  & -0.478 & 0.342 \\
 &  & (4.803) & (3.834) &  & (1.515) & (1.189) \\
 &  &  &  &  &  &  \\
Observations & 69 & 69 & 69 & 69 & 69 & 69 \\
 Control Group Mean & 4.840 & 4.840 & 8.828 & 1.101 & 1.101 & 2.433 \\ \hline
\end{tabular}

\begin{tablenotes}Notes: This table uses data from the Phase 2 experimental and network dataset. The table presents OLS regressions of number of calls received (and number of calls received normalized by the number of seeds, 3 or 5, which is randomly assigned) on characteristics of the set of seeds. High $DC$ refers to a seed being above the mean by one standard deviation of the centrality distribution. All columns control for total number of gossips, number of elders, and number of seeds. For columns (1), (2), (4), and (5), the control group mean is calculated as the mean expectation of the outcome variable when no gossips or elders are reached. For columns (3) and (6), the control group mean is calculated as the mean expectation of the outcome variable when no high $DC$ seeds are reached. Robust standard errors are reported in parentheses.  
\end{tablenotes}
\end{threeparttable}}
\end{table}

\clearpage

\end{document}